\pdfoutput=1
\documentclass[11pt]{article}
\usepackage{times}
\usepackage[letterpaper, left=1in, right=1in, top=1in, bottom=1in]{geometry}
\usepackage{mathtools}

\DeclareMathOperator*{\argmin}{argmin}

\newcommand{\field}[1]{\mathbb{#1}}

\newcommand{\E}{\field{E}}

\newcommand{\KL}{{\text{\rm KL}}}

\newcommand{\Reg}{\mathrm{Reg}}

\newcommand{\wh}{\widehat}
\newcommand{\wt}{\widetilde}

\newcommand{\R}{\mathbb{R}}
\newcommand{\Xcal}{\mathcal{X}}
\newcommand{\Acal}{\mathcal{A}}

\newcommand{\otil}{\ensuremath{\widetilde{\mathcal{O}}}}

\newcommand{\Dcal}{\mathcal{D}}
\newcommand{\Ecal}{\mathcal{E}}

\newcommand{\Gcal}{\mathcal{G}}

\newcommand{\Ocal}{\mathcal{O}}

\newcommand{\EE}{\mathbb{E}} 
\newcommand{\pa}[1]{\left(#1\right)}

\newcommand{\piref}{\pi_{\textrm{ref}}}

\newcommand{\Softmax}{\operatorname{Softmax}}

\DeclareMathOperator{\supp}{supp}

\usepackage{microtype}
\usepackage{graphicx}
\usepackage{subfigure}
\usepackage{booktabs} 

\usepackage{enumitem}

\usepackage{hyperref}

\usepackage{amsmath}
\usepackage{amssymb}
\usepackage{mathtools}
\usepackage{amsthm}
\usepackage{mathrsfs} 
\usepackage{algorithm}
\usepackage{algorithmic}
\usepackage{thmtools,thm-restate}
\allowdisplaybreaks

\usepackage[capitalize,noabbrev]{cleveref}
\theoremstyle{plain}
\newtheorem{theorem}{Theorem}

\newtheorem{lemma}[theorem]{Lemma}
\newtheorem{corollary}[theorem]{Corollary}

\theoremstyle{definition}
\newtheorem{definition}{Definition}
\newtheorem{assumption}[definition]{Assumption}
\theoremstyle{remark}

\usepackage{natbib}
\date{}

\usepackage[textsize=tiny]{todonotes}

\usepackage{hyperref,url}
\hypersetup{
  colorlinks,
  citecolor=blue!70!black,
  linkcolor=blue!70!black,
  urlcolor=blue!70!black}
\usepackage{cleveref}
\crefformat{equation}{#2(#1)#3}
\Crefformat{equation}{#2(#1)#3}

\def\shownotes{1}  
\ifnum\shownotes=1
\newcommand{\authnote}[2]{{\scriptsize $\ll$\textsf{#1 notes: #2}$\gg$}}
\else
\newcommand{\authnote}[2]{}
\fi

\title{Beyond Pessimism: Offline Learning in KL-regularized Games}

\author{Yuheng Zhang\thanks{University of Illinois Urbana-Champaign. Email: \texttt{yuhengz2@illinois.edu, nanjiang@illinois.edu}.} \and Claire Chen\thanks{California Institute of Technology. Email: \texttt{clairechen@caltech.edu}.}  \and Nan Jiang\footnotemark[1]}

\renewcommand{\cite}{\citep}

\begin{document}
\maketitle

\begin{abstract}
We study offline learning in KL-regularized two-player zero-sum games, where policies are optimized with respect to a fixed reference policy through KL regularization.
Prior work relies on pessimistic value estimation to handle distribution shift, yielding only $\widetilde{\mathcal{O}}(1/\sqrt n)$ statistical rates.
We develop a new pessimism-free algorithm and analytical framework for KL-regularized games, built on the smoothness of KL-regularized best responses and a stability property of the Nash equilibrium induced by skew symmetry.
This yields, to our knowledge, the first pessimism-free offline learning guarantee for KL-regularized games, with a fast $\widetilde{\mathcal{O}}(1/n)$ sample complexity bound.
We further propose an efficient self-play policy optimization algorithm that replaces exact equilibrium computation with iterative KL-regularized policy updates, and prove that its last iterate preserves the same pessimism-free statistical guarantee up to a controlled optimization error.
\end{abstract}

\section{Introduction}
Offline reinforcement learning (RL) aims to learn decision-making policies from a fixed dataset, without any additional interaction with the underlying environment.
This paradigm is particularly compelling in settings where online exploration is costly or unsafe, such as healthcare \citep{komorowski2018artificial}, finance \citep{lee2023offline}, and autonomous driving \citep{shi2021offline}.
As a result, understanding the statistical limits and designing provably reliable offline learning algorithms have become central topics in data-driven decision-making.

A fundamental difficulty in offline learning is distribution shift: the policies being optimized may differ substantially from the behavior policy that generated the data. As a result, value or payoff estimates can become unreliable when evaluating actions that are rarely observed in the dataset. To address this issue, most existing approaches to offline RL and offline game learning rely on the principle of pessimism \citep{liu2020provably,jin2021pessimism,rashidinejad2021bridging,cui2022offline,zhang2023offline,ye2024online}. These methods modify the learning objective by introducing conservative penalties or uncertainty adjustments, so that actions outside the data distribution are assigned low values. While effective for controlling extrapolation error, such pessimistic designs typically introduce additional algorithmic complexity and require careful tuning in practice.

In this paper, we study offline learning in KL-regularized two-player zero-sum games. In this formulation, each player seeks to optimize its own objective while being constrained, via a KL divergence, to remain close to a fixed reference policy. This regularization induces a unique Nash equilibrium and provides a principled way to control deviations from a baseline. Crucially, such games serve as the theoretical backbone for Large Language Model (LLM) alignment, as seen in Nash learning from human feedback~\citep{munos2024nash}, where policies are optimized via pairwise comparisons under enforced KL regularization. Prior work on offline learning in KL-regularized games largely follows the pessimism-based paradigm developed for offline Markov games. In particular, \citet{ye2024online} applies pessimistic value estimation to KL-regularized zero-sum games, obtaining a $\widetilde{\mathcal{O}}(1/\sqrt{n})$ sample complexity bound under unilateral coverage. While this approach provides statistical guarantees, it does not fully exploit the geometric structure induced by KL regularization, notably the strong convexity and smoothness of the regularized game. This raises a natural question: 
\begin{center}
\textit{Can we achieve sharper statistical guarantees in KL-regularized games without pessimism?}
\end{center}
We provide an affirmative answer to this question. To the best of our knowledge, this is the first work to demonstrate that by properly exploiting the geometry of the problem, one can achieve a fast $\widetilde{\mathcal{O}}(1/n)$ rate using a completely pessimism-free algorithm. Our main contributions are:
\begin{itemize}\item \textbf{Pessimism-free Algorithm with Fast Rates.} We propose a direct minimax estimation algorithm (Algorithm~\ref{alg:minimax}) that performs equilibrium computation on the empirical game induced by a simple least-squares payoff estimator. We prove that under the standard unilateral coverage assumption, this approach achieves a $\widetilde{\mathcal{O}}(1/n)$ sample complexity, significantly improving upon the pessimism-based $\widetilde{\mathcal{O}}(1/\sqrt{n})$ rates established in prior work \citep{ye2024online}. This result provides the first proof that explicit pessimism is \textit{not} necessary for sample-efficient offline learning in this setting.

\item \textbf{Novel Geometric Analysis.} Our analysis fundamentally departs from the standard pessimism-based proofs. Instead of constructing lower confidence bounds, we leverage the stability of KL-regularized equilibria. By exploiting the smoothness of best responses and the skew-symmetry of the zero-sum game operator, we show that the duality gap is controlled solely by unilateral estimation errors. This property allows us to harness the fast generalization rates of least-squares estimation directly, without requiring pessimism or uniform policy coverage.

\item \textbf{Efficient Self-play Optimization.} To address the computational bottleneck of exact equilibrium computation, we propose an efficient self-play policy optimization algorithm (Algorithm~\ref{alg:self_play}). We analyze the coupled mirror descent-ascent dynamics and show that after $T$ iterations, the self-play algorithm matches the $\widetilde{\mathcal{O}}(1/n)$ statistical guarantee of the minimax estimator up to a vanishing $\mathcal{O}(1/\sqrt{T})$ optimization error, thereby providing both statistical guarantees and computational efficiency.
\end{itemize}

\section{Preliminary}\label{sec:prelim}
\paragraph{Notations.}
Let $\Xcal$ be the context space with $x \sim \rho$, where $\rho \in \Delta(\Xcal)$ is the context distribution.
Both players share a finite action space $\Acal$.
A policy for player $i\in\{1,2\}$ is denoted by $\pi_i:\Xcal \to \Delta(\Acal)$, which specifies the action distribution given the context.
A joint policy is written as $\pi=(\pi_1,\pi_2)$. We define the $\ell_1$-norm of a joint policy as the sum of the $\ell_1$-norms of its components, i.e., $\|\pi\|_1 \coloneqq \|\pi_1\|_1 + \|\pi_2\|_1$. For a distribution $\mu\in\Delta(\Acal)$, we write
$\supp(\mu):=\{a\in\Acal:\mu(a)>0\}$ for its support.


\paragraph{KL-regularized Game Objective.} We consider a two-player zero-sum game with an unknown payoff function
$g^\star:\Xcal \times \Acal \times \Acal \to [-1,1]$.
Given a reference policy $\piref$ and a regularization parameter $\eta$, we define the KL-regularized game objective as
\begin{align}
\textstyle    J(g^\star,\pi_1,\pi_2)
    := 
    \EE_{x \sim \rho}
    \EE_{\substack{a_1 \sim \pi_1(\cdot|x)\\ a_2 \sim \pi_2(\cdot|x)}}
    \bigg[
        g^\star(x,a_1,a_2)
        - \eta^{-1} \log \frac{\pi_1(a_1|x)}{\piref(a_1|x)}
        + \eta^{-1} \log \frac{\pi_2(a_2|x)}{\piref(a_2|x)}
    \bigg].
    \label{eq:kl-objective-game} \nonumber
\end{align}
Player~1 seeks to maximize this objective, while Player~2 seeks to minimize it. The KL regularization terms encourage both players to remain close to the reference policy $\piref$. Such KL-regularized games naturally arise in RLHF and alignment for large language models~\citep{ouyang2022training, munos2024nash}, where $g^\star$ can represent the win rate between two candidate responses.
\paragraph{Nash Equilibrium and Duality Gap.}
We restrict attention to the policy class $\Pi$ consisting of all policies supported on the same action set as the reference policy $\piref$.
For a payoff function $g$, we use the shorthand:
\[
J(g,\dagger,\pi_2)
~:=~
\max_{\pi_1\in\Pi} J(g,\pi_1,\pi_2),
\qquad
J(g,\pi_1,\dagger)
~:=~
\min_{\pi_2\in\Pi} J(g,\pi_1,\pi_2),
\]
to denote best response values. The Nash equilibrium (NE) of the KL-regularized game is defined as a saddle point:
\begin{align*}
(\pi_1^\star,\pi_2^\star)
~\in~
\arg\max_{\pi_1 \in \Pi}\;
\arg\min_{\pi_2 \in \Pi}
J(g^\star,\pi_1,\pi_2).
\end{align*}
Due to the strong convexity and concavity induced by the KL regularization, this equilibrium is unique~\citep{munos2024nash, ye2024online}. Moreover, $(\pi_1^\star,\pi_2^\star)$ satisfies the mutual best response conditions $J(g^\star,\pi_1^\star,\pi_2^\star)=J(g^\star,\dagger,\pi_2^\star)=J(g^\star,\pi_1^\star,\dagger)$. For any policy pair $\pi=(\pi_1,\pi_2)$, we quantify suboptimality via the duality gap:
\begin{align*}
\mathrm{DualGap}(\pi)
~:=~
J(g^\star,\dagger,\pi_2)
-
J(g^\star,\pi_1,\dagger).
\end{align*}
Note that $\mathrm{DualGap}(\pi)\ge 0$ for all $\pi$, and $\mathrm{DualGap}(\pi)=0$ if and only if $\pi$ is a Nash equilibrium.

\paragraph{Function Approximation.}
We consider a general function approximation setting with a function class
$\Gcal \subset (\Xcal \times \Acal \times \Acal \to [-1,1])$ for estimating the payoff function $g^\star$.
We make the following realizability assumption, which is standard in function approximation literature~\citep{xie2021batch,zhang2023offline}.

\begin{assumption}[Realizability]\label{asm:realize}
The function class $\Gcal$ is finite and contains the true payoff function, i.e., $g^\star \in \Gcal$.
\end{assumption}
The finite class assumption is made for notational simplicity and can be relaxed to standard complexity measures such as covering numbers.

\paragraph{Offline Dataset.}
In the offline setting, we do not interact with the environment and instead assume access to a pre-collected dataset
$\Dcal=\{(x_i,a_{i,1},a_{i,2},p_i)\}_{i=1}^n$.
Each context $x_i$ is drawn from $\rho$, and the action pair $(a_{i,1}, a_{i,2})$ is sampled from a joint behavior policy $\pi_b = (\pi_{b,1}, \pi_{b,2})$, i.e., 
$a_{i,1} \sim \pi_{b,1}(\cdot \mid x_i)$ and $a_{i,2} \sim \pi_{b,2}(\cdot \mid x_i)$. The observed feedback is
\[
p_i = g^\star(x_i,a_{i,1},a_{i,2}) + \epsilon_i,
\]
where $\epsilon_i$ is $1$-sub-Gaussian. For example, in LLM alignment, $x_i$ is a user prompt, $(a_{i,1},a_{i,2})$ are two candidate responses, and $p_i$ is a (noisy) human preference.

\paragraph{Concentrability.}
The coverage of the offline dataset $\Dcal$ plays a central role in determining what policies can be learned effectively. 
Since learning relies only on data collected under $\pi_b$, generalization requires that errors measured under $\rho \times \pi_b$ remain controlled on context-action tuples induced by other policies. 
To quantify this form of distribution mismatch, we adopt the $D^2$-divergence notion from prior work on offline learning~\citep{ye2024online,zhao2025towards}.
\begin{definition}
Given a function class $\mathcal{G} \subset (\Xcal \times \Acal \times \Acal \to [-1,1])$ and a behavior policy $\pi_b$, the $D^2$-divergence at a triple $(x,a_1,a_2)$ is defined as
$$
    D^2_{\mathcal{G}}((x,a_1,a_2);\pi_b):= \sup_{g,h \in \mathcal{G}}
    \frac{\big( g(x,a_1,a_2) - h(x,a_1,a_2) \big)^2}
    {\E_{(x',a'_1,a'_2) \sim \rho \times \pi_b}
    \big[ \big( g(x',a'_1,a'_2) - h(x',a'_1,a'_2) \big)^2 \big]}.
$$
\end{definition}
The $D^2$-divergence acts as a pointwise extrapolation factor: it quantifies how much
a discrepancy between two functions in $\mathcal{G}$, measured under the data-generating distribution $\rho \times \pi_b$, can be amplified at a particular context--action triple $(x,a_1,a_2)$. 
A small value of $D^2_{\mathcal{G}}((x,a_1,a_2);\pi_b)$ guarantees that if two functions are close under the data distribution, they must also be close at $(x,a_1,a_2)$.

With the $D^2$-divergence in hand, we are ready to state our coverage assumption on $\Dcal$.
Unlike offline learning in single-agent problems, where it suffices to assume coverage of the optimal policy, provably efficient offline learning in games requires \emph{unilateral coverage}: the dataset must cover not only the Nash policy pair $(\pi_1^\star,\pi_2^\star)$, but also all unilateral policy pairs obtained by allowing one player to deviate while the other plays its Nash strategy.

\begin{assumption}[Unilateral Concentrability]\label{assm:uni_cover}
There exists a constant $C_{\mathrm{uni}} < \infty$ such that for both players
$i \in \{1,2\}$,
\[
\sup_{\pi_i}
\EE_{(x,a_1,a_2)\sim \rho \times \pi_i \times \pi^\star_{-i}}
\!\left[
D^2_{\mathcal{G}}\big((x,a_1,a_2);\pi_b\big)
\right]
\;\le\;
C_{\mathrm{uni}},
\]
where $\pi^\star_{-i}$ denotes the Nash policy of the opponent.
\end{assumption}
This condition requires the behavior policy $\pi_b$ to cover all \emph{unilateral deviation distributions} from the NE. As shown in~\citet{cui2022offline}, such unilateral coverage is in fact necessary for solving offline two-player zero-sum games.

\section{Offline Learning in KL-Regularized Games without Pessimism}\label{sec:info}
In this section, we study the sample complexity of offline learning in KL-regularized contextual zero-sum games under the unilateral coverage condition, focusing on statistical rather than computational efficiency. We show that  sharp sample complexity guarantees can be obtained without introducing any form of pessimism. These results characterize the statistical difficulty of the problem and will serve as the foundation for the efficient algorithm developed in the next section.

\subsection{Algorithm}
\begin{algorithm*}[t]
\caption{Offline Learning in KL-Regularized Games without Pessimism}
\label{alg:minimax}
\begin{algorithmic}[1]

\STATE \textbf{Input:} coefficient $\eta>0$, offline dataset 
$\mathcal{D}$, reference policy $\piref$, function class $\mathcal{G}$.

\STATE \textbf{Payoff estimation:} compute the least-squares estimator
\[
\wh{g} \in \arg\min_{g\in\mathcal{G}}
\sum_{i=1}^n \big(g(x_i,a_{i,1},a_{i,2})-p_i\big)^2 .
\]

\STATE \textbf{Equilibrium computation:} solve the KL-regularized zero-sum game induced by $\wh{g}$,
\[
\wh{\pi}\in 
\arg\max_{\pi_1\in\Pi}\;\arg\min_{\pi_2\in\Pi}\; J(\wh{g},\pi_1,\pi_2) .
\]

\STATE \textbf{Output:} policy pair $\wh{\pi}=(\wh{\pi}_1,\wh{\pi}_2)$.
\end{algorithmic}
\end{algorithm*}

The algorithm is summarized in Algorithm~\ref{alg:minimax}.
A key distinction from existing offline equilibrium learning methods is that our approach does \emph{not} employ any form of pessimism in payoff estimation.
In prior work on offline learning in games, pessimism is typically enforced either by subtracting pointwise uncertainty bonuses from the estimated payoff or by taking worst-case values over a constructed version space~\citep{cui2022offline,zhang2023offline,ye2024online}. Such pessimistic designs often introduce additional algorithmic and implementation complexity, as they require tuning uncertainty terms or solving auxiliary optimization problems. In contrast, our algorithm relies solely on least-squares estimation without introducing any explicit pessimistic bias.

\subsection{Analysis}
Pessimism-based algorithms and analyses are the standard approach for obtaining sample complexity guarantees in offline RL and offline game learning~\citep{liu2020provably,jin2021pessimism,cui2022offline}. 
In this subsection, we present a novel analysis that departs from this paradigm. Note that the duality gap of a policy pair $\pi=(\pi_1,\pi_2)$ can be written as
\begin{align*}
\mathrm{DualGap}(\pi)=\pa{J(g^\star,\dagger,\pi_2)
-J(g^\star,\pi_1^\star,\pi_2^\star)}+
\pa{J(g^\star,\pi_1^\star,\pi_2^\star)-
J(g^\star,\pi_1,\dagger)}.
\end{align*}
By symmetry of the two players, it suffices to control one of the two terms above.

\subsubsection{Previous Analysis in Offline Game Learning}
Before proceeding to our analysis, we briefly review the pessimism-based framework adopted in prior work on offline game learning~\citep{ye2024online}. 
Their algorithm constructs a pessimistic payoff estimate $\underline g$ such that
$J(\underline g,\cdot,\cdot)$ is a high-probability lower bound of $J(g^\star,\cdot,\cdot)$.

Let $\wh{\pi}_1$ denote the output policy of the pessimistic algorithm, and define the best responses under the true and pessimistic games by
$\pi_2^\dagger=\argmin_{\pi_2\in\Pi} J(g^\star,\wh{\pi}_1,\pi_2)$ and $\wt{\pi}_2=\argmin_{\pi_2\in\Pi} J(\underline g,\wh{\pi}_1,\pi_2)$,
and similarly
\(
\wt{\pi}_2^\star=\argmin_{\pi_2\in\Pi} J(\underline g,\pi_1^\star,\pi_2).
\)
The suboptimality of $\wh{\pi}_1$ can then be decomposed as
\begin{align*}
&J(g^\star,\pi_1^\star,\pi_2^\star)-J(g^\star,\wh{\pi}_1,\pi_2^\dagger)
=
\underbrace{J(g^\star,\pi_1^\star,\pi_2^\star)-J(g^\star,\pi_1^\star,\wt{\pi}_2^\star)}_{q_1}
+\underbrace{J(g^\star,\pi_1^\star,\wt{\pi}_2^\star)-J(\underline g,\pi_1^\star,\wt{\pi}_2^\star)}_{q_2}\\
&+\underbrace{J(\underline g,\pi_1^\star,\wt{\pi}_2^\star)-J(\underline g,\wh{\pi}_1,\wt{\pi}_2)}_{q_3}
+\underbrace{J(\underline g,\wh{\pi}_1,\wt{\pi}_2)-J(\underline g,\wh{\pi}_1,\pi_2^\dagger)}_{q_4}+\underbrace{J(\underline g,\wh{\pi}_1,\pi_2^\dagger)-J(g^\star,\wh{\pi}_1,\pi_2^\dagger)}_{q_5}.
\end{align*}

By the Nash optimality of $(\pi_1^\star,\pi_2^\star)$ under $J$ and the saddle-point property of $(\wh{\pi}_1,\wt{\pi}_2)$ under $\underline g$, we have $q_1\le 0$ and $q_3\le 0$.
Moreover, since $\wt{\pi}_2$ is the best response to $\wh{\pi}_1$ under $\underline g$, $q_4\le 0$.
Finally, the pessimism construction guarantees $\underline g \le g^\star$, which implies $q_5\le 0$.
Therefore, the entire gap reduces to the single statistical term
\[
q_2
=
J(g^\star,\pi_1^\star,\wt{\pi}_2^\star)-J(\underline g,\pi_1^\star,\wt{\pi}_2^\star),
\]
which can be bounded by $\otil(1/\sqrt{n})$ via standard concentration arguments. 
This pessimism-based decomposition is the standard route to obtain sample complexity guarantees in offline game learning. However, it typically yields only a $\otil(1/\sqrt{n})$ statistical rate and fundamentally relies on constructing a pessimistic surrogate, which is precisely what we avoid in our analysis.

\subsubsection{Analysis without Pessimism: A Fast-rate Bound}
In this subsection, we develop an analysis for Algorithm~\ref{alg:minimax} that is fundamentally different from the pessimism-based approaches in prior offline game learning. By leveraging the geometric structure induced by KL regularization, we establish an $\otil(1/n)$ sample complexity bound, significantly improving upon the prior $\otil(1/\sqrt{n})$ rate. Recall that it suffices to control one side of the duality gap $J(g^\star,\dagger,\wh{\pi}_2)-J(g^\star,\pi_1^\star,\pi_2^\star)$. We decompose the term as
\[
J(g^\star,\dagger,\wh{\pi}_2)-J(g^\star,\pi_1^\star,\pi_2^\star)
\!=\!
\underbrace{J(g^\star,\dagger,\wh{\pi}_2)\!-\!J(g^\star,\pi_1^\star,\wh{\pi}_2)}_{\mathrm{Gap}_1}
\!+\!
\underbrace{J(g^\star,\pi_1^\star,\wh{\pi}_2)\!-\!J(g^\star,\pi_1^\star,\pi_2^\star)}_{\mathrm{Gap}_2}.
\]
At a high level, our analysis proceeds in three steps. First, we relate $\mathrm{Gap}_1$ and $\mathrm{Gap}_2$ to the policy deviation between $\wh{\pi}$ and $\pi^\star$, together with a unilateral estimation error term. Second, we invoke a stability argument for the KL-regularized game to bound the policy distance by the unilateral estimation error. Finally, we combine these results with the unilateral concentrability assumption to obtain the fast rate. The full proofs for the lemmas in this subsection are deferred to Appendix~\ref{app:proof_minimax}.

To begin with, we present Lemma~\ref{lem:kl_subopt} for KL-regularized games that characterizes how the value loss incurred by deviating from a best response can be expressed in terms of a KL divergence. 
\begin{lemma}[Best Response Suboptimality as a KL Divergence]
\label{lem:kl_subopt}
Fix any $\pi_1\in\Pi$ and let the Player~2 best response be
\(
\pi_2^\dagger(\pi_1)\in\arg\min_{\pi_2\in\Pi} J(g^\star,\pi_1,\pi_2).
\)
Then for any $\pi_2\in\Pi$,
\[
J(g^\star,\pi_1,\pi_2)-J(g^\star,\pi_1,\pi_2^\dagger(\pi_1))
=\eta^{-1}\E_{x\sim\rho}\!\left[\KL\!\big(\pi_2(\cdot|x)\,\|\,\pi_2^\dagger(\pi_1)(\cdot|x)\big)\right].
\]
Similarly, fix any $\pi_2\in\Pi$ and let the Player~1 best response be
\(
\pi_1^\dagger(\pi_2)\in\arg\max_{\pi_1\in\Pi} J(g^\star,\pi_1,\pi_2).
\)
Then for any $\pi_1\in\Pi$,
\[
J(g^\star,\pi_1^\dagger(\pi_2),\pi_2)-J(g^\star,\pi_1,\pi_2)
=\eta^{-1}\E_{x\sim\rho}\!\left[\KL\!\big(\pi_1(\cdot|x)\,\|\,\pi_1^\dagger(\pi_2)(\cdot|x)\big)\right].
\]
\end{lemma}
This lemma allows us to translate suboptimality with respect to best responses into a KL discrepancy
between policies, which will play a central role in the subsequent analysis. By Lemma~\ref{lem:kl_subopt}, fixing $\pi_2=\wh{\pi}_2$ and taking $\pi_1=\pi_1^\star$, we have
\begin{align}\label{eq:gap1_kl}
\mathrm{Gap}_1
= J(g^\star,\dagger,\wh{\pi}_2)-J(g^\star,\pi_1^\star,\wh{\pi}_2)
= \eta^{-1}\E_{x\sim\rho}\!\left[\KL\!\big(\pi_1^\star(\cdot|x)\,\|\,\pi_1^\dagger(\wh{\pi}_2)(\cdot|x)\big)\right].
\end{align}
Similarly, fixing $\pi_1=\pi_1^\star$ and taking $\pi_2=\wh{\pi}_2$, we obtain
\begin{align}\label{eq:gap2_kl}
\mathrm{Gap}_2
= J(g^\star,\pi_1^\star,\wh{\pi}_2)-J(g^\star,\pi_1^\star,\pi_2^\star)
= \eta^{-1}\E_{x\sim\rho}\!\left[\KL\!\big(\wh{\pi}_2(\cdot|x)\,\|\,\pi_2^\star(\cdot|x)\big)\right].
\end{align}
\paragraph{Geometry of KL-regularized Best Responses.} To relate these KL terms to the policy distance, we exploit the softmax structure of KL-regularized best responses. Specifically, the KL divergence between two best response policies can be controlled by the difference between their logits, which depend linearly on the opponent’s policy. We formalize this relationship in the following lemma, a canonical result in the online learning and information-theoretic literature (e.g., \citet{cesa2006prediction,cover1999elements}).


\begin{lemma}[KL Bound via Logits Distance]
\label{lem:kl_logit}
Let $P=\Softmax(z)$ and $Q=\Softmax(z')$ for some logits $z,z'\in\R^{|\Acal|}$. Then
\[
\KL(P\|Q) \le \frac{1}{2}\|z-z'\|_\infty^2 .
\]
\end{lemma}
Lemma~\ref{lem:kl_logit} allows us to upper bound the KL divergence between two softmax distributions by the squared $\ell_\infty$ distance between their logits. We apply this lemma first to bound $\mathrm{Gap}_1$. Fix a context $x$. Since both $\pi_1^\dagger(\wh{\pi}_2)(\cdot|x)$ and $\pi_1^\star(\cdot|x)$ are KL-regularized best responses, they admit softmax representations
$\pi_1^\dagger(\wh{\pi}_2)(\cdot|x)=\Softmax(z^\dagger(x))$ and $\pi_1^\star(\cdot|x)=\Softmax(z^\star(x))$,
with logits
$z^\dagger_a(x)=\eta\,\E_{a_2\sim\wh{\pi}_2}[g^\star(x,a,a_2)]+\log\piref(a|x)$ and
$z^\star_a(x)=\eta\,\E_{a_2\sim\pi_2^\star}[g^\star(x,a,a_2)]+\log\piref(a|x)$. Using Hölder's inequality and the fact that $|g^\star|\le 1$, the $\ell_\infty$-distance is bounded by
\[
\|z^\dagger(x)-z^\star(x)\|_\infty
\le \eta\,\|\wh{\pi}_2(\cdot|x)-\pi_2^\star(\cdot|x)\|_1 .
\]
By Lemma~\ref{lem:kl_logit}, this implies
\[
\KL\big(\pi_1^\star(\cdot|x)\,\|\,\pi_1^\dagger(\wh{\pi}_2)(\cdot|x)\big)
\le \tfrac{\eta^2}{2}\,\|\wh{\pi}_2(\cdot|x)-\pi_2^\star(\cdot|x)\|_1^2 .
\]
Taking expectation over $x\sim\rho$ and using Eq.~\eqref{eq:gap1_kl}, we obtain
\[
\mathrm{Gap}_1 \le \frac{\eta}{2}\,\E_{x \sim \rho}\!\big[\|\wh{\pi}_2(\cdot|x)-\pi_2^\star(\cdot|x)\|_1^2\big].
\]
On the other hand, for $\mathrm{Gap}_2$, we must address the mismatch between the estimated payoff $\wh{g}$ and the true payoff $g^\star$. The logits of $\wh{\pi}_2$ and $\pi_2^\star$ differ by $-\eta (\E_{\wh{\pi}_1}[\wh{g}] - \E_{\pi_1^\star}[g^\star])$. To invoke the unilateral concentrability assumption, we decompose this difference by adding and subtracting $\E_{a_1 \sim \pi_1^\star(\cdot|x)}[\wh{g}]$, which isolates the estimation error on the Nash policy:
\[
\wh{z}_a(x) - z^\star_a(x)
= -\eta \Big(
\underbrace{\E_{a_1 \sim \wh{\pi}_1(\cdot|x)}[\wh{g}] - \E_{a_1 \sim \pi_1^\star(\cdot|x)}[\wh{g}]}_{\text{Policy Deviation}}
+
\underbrace{\E_{a_1 \sim \pi_1^\star(\cdot|x)}[\wh{g} - g^\star]}_{\text{Unilateral Estimation Error}}
\Big).
\]
For the policy deviation term, Hölder's inequality implies it is bounded by $\|\wh{\pi}_1(\cdot|x) - \pi_1^\star(\cdot|x)\|_1$. Consequently, applying Lemma~\ref{lem:kl_logit} to the decomposed logits yields the following bound:
\begin{align}\label{eq:gap2_bound}
\mathrm{Gap}_2
\!\le\!
2\eta\E_{x\sim\rho}\big[ \|\wh{\pi}_1(\cdot|x) - \pi_1^\star(\cdot|x)\|_1^2 \big]
\!+\!
2\eta\E_{x\sim\rho}\Big[ \!\max_{a_2 \in \Acal} \big( \E_{a_1 \sim \pi_1^\star(\cdot|x)}[\wh{g}(x,a_1,a_2) - g^\star(x,a_1,a_2)] \!\big)^2\! \Big].
\end{align}
The second term in Eq.~\eqref{eq:gap2_bound} is exactly the squared unilateral estimation error, which is controlled by our data coverage assumption. The first term, together with the bound for $\mathrm{Gap}_1$, constitutes the total policy distance from the Nash equilibrium, which is analyzed via the following stability argument.

\paragraph{Stability under Unilateral Errors.} 
Fix a context $x$ and denote
$\pi_x\!=\!(\!\pi_1(\cdot|x),\pi_2(\cdot|x)\!)$.
Let $G_g(x)\in\R^{A\times A}$ be the payoff matrix with entries
$[G_g(x)]_{a_1,a_2}=g(x,a_1,a_2)$. We define the KL-regularized saddle-point operator at context $x$ by
\[
\textstyle F_{g,x}(\pi_x):=M_{g,x}\pi_x+\eta^{-1}\nabla\Reg(\pi_x),
\qquad
M_{g,x}:=\begin{pmatrix}0&-G_g(x)\\ G_g(x)^\top&0\end{pmatrix},
\]
where $\Reg(\pi_x)
=
\KL(\pi_1(\cdot|x)\,\|\,\piref(\cdot|x))
+
\KL(\pi_2(\cdot|x)\,\|\,\piref(\cdot|x)).$ The Nash equilibrium satisfies the corresponding first-order optimality condition, which we write in operator form as $F_{g^\star,x}(\pi_x^\star)=0$ and $F_{\wh g,x}(\wh\pi_x)=0$.

The KL regularizer is strongly convex in the joint $\ell_1$ geometry, which implies that the saddle-point operator $F_{g,x}$ is $\frac{1}{2\eta}$-strongly monotone. Therefore,
\[
\textstyle\frac{1}{2\eta}\,\|\wh{\pi}_x-\pi_x^\star\|_1^2
\le
\big\langle
F_{g^\star,x}(\wh{\pi}_x)-F_{g^\star,x}(\pi_x^\star),
\wh{\pi}_x-\pi_x^\star
\big\rangle .
\]
Using the identities $F_{g^\star,x}(\pi_x^\star)=0$ and $F_{\wh{g},x}(\wh{\pi}_x)=0$, and noting that the regularization terms cancel, the RHS reduces to
$\big\langle (M_{g^\star,x}-M_{\wh{g},x})\wh{\pi}_x,\wh{\pi}_x-\pi_x^\star\big\rangle$. 

Next we use the skew-symmetry of the payoff operator $M_{g,x}$, a known property
of two-player zero-sum games when written in operator form \citep{nemirovski2004prox,rakhlin2013optimization},
and decompose $\wh{\pi}_x=\pi_x^\star+(\wh{\pi}_x-\pi_x^\star)$.
Since $\Delta M_x \coloneqq M_{g^\star,x}-M_{\wh{g},x}$ is skew-symmetric, the quadratic term 
$\langle \Delta M_x(\wh{\pi}_x-\pi_x^\star),\,\wh{\pi}_x-\pi_x^\star\rangle$ vanishes, 
and the right-hand side reduces to 
$\langle \Delta M_x\pi_x^\star,\,\wh{\pi}_x-\pi_x^\star\rangle$. 
This cancellation localizes the analysis to the estimation error evaluated only
on unilateral deviations from $\pi^\star$. We summarize this stability
result in Lemma~\ref{lem:l1_stability}.
\begin{lemma}[Stability under Unilateral Estimation Error]
\label{lem:l1_stability}
Define the unilateral estimation error at context $x$ by
\[
\mathcal{E}(x)
\!\!\;:=\;\!\!
\max\!\Big\{
\!\big\|\E_{a_2\sim \pi_2^\star(\cdot|x)}[\wh{g}(x,\cdot,a_2)-g^\star(x,\cdot,a_2)]\big\|_\infty,\;
\!\big\|\E_{a_1\sim \pi_1^\star(\cdot|x)}[\wh{g}(x,a_1,\cdot)-g^\star(x,a_1,\cdot)]\big\|_\infty\!
\Big\}.
\]
For any context $x$, let $\pi_x^\star$ and $\wh{\pi}_x$ denote the Nash equilibria of $g^\star$ and $\wh{g}$, respectively. Then
\[
\|\wh{\pi}_x-\pi_x^\star\|_1 \le 2\eta\,\mathcal{E}(x).
\]
\end{lemma}
Lemma~\ref{lem:l1_stability} implies that the learned policy $\wh{\pi}$ remains close to the true Nash equilibrium $\pi^\star$ as long as the payoff estimation error is small under unilateral deviations. This stability result allows us to convert the policy distances $\|\wh{\pi} - \pi^\star\|_1$ into estimation errors evaluated only on the unilateral deviations from $\pi^\star$, thereby bypassing the need for coverage on the learned policy $\wh{\pi}$.

\paragraph{Unilateral Concentrability.}
We now combine the bounds for $\mathrm{Gap}_1$ and $\mathrm{Gap}_2$ to derive the final sample complexity.
Substituting the stability bound from Lemma~\ref{lem:l1_stability} into Eq.~\eqref{eq:gap1_kl} and Eq.~\eqref{eq:gap2_bound},
\begin{align*}
\mathrm{Gap}_1 + \mathrm{Gap}_2 
\le (2\eta + 10\eta^3) \E_{x\sim\rho}\big[\mathcal{E}(x)^2\big].
\end{align*}
By the definition of the $D^2$-divergence and Assumption~\ref{assm:uni_cover}, for any unilateral deviation distribution $\mu=\rho\times \pi_i\times \pi_{-i}^\star$ with $i\in\{1,2\}$, we have
\[
\E_{\mu}\!\big[(\wh{g}-g^\star)^2\big]
\;\le\;
\E_{\mu}\!\big[D_{\Gcal}^2(\cdot;\pi_b)\big]\,
\E_{\rho\times \pi_b}\!\big[(\wh{g}-g^\star)^2\big]
\;\le\;
C_{\mathrm{uni}}\,
\E_{\rho\times \pi_b}\!\big[(\wh{g}-g^\star)^2\big].
\]
Since $\wh{g}$ is obtained by least squares over the finite class $\Gcal$, standard fast-rate generalization bounds imply that with probability at least $1-\delta$, the squared error on the training distribution scales as $\otil(1/n)$.
Combining these results, we conclude that
\[
\mathrm{DualGap}(\wh{\pi}) \le (2\eta + 10\eta^3) \cdot C_{\mathrm{uni}} \cdot \otil\left(\frac{1}{n}\right) = \otil\left(\frac{(\eta+\eta^3) C_{\mathrm{uni}}}{n}\right),
\]
which completes the analysis.

\subsection{Theoretical Guarantee}
\begin{theorem}[Fast-rate Sample Complexity Bound without Pessimism]
\label{thm:fast_rate_sc}
Under Assumptions~\ref{asm:realize} and~\ref{assm:uni_cover}, let $\wh{\pi}$ be the output of Algorithm~\ref{alg:minimax}, with probability at least $1-\delta$,
$$
\mathrm{DualGap}(\wh{\pi}) \le \Ocal\pa{\frac{(\eta+\eta^3)\,C_{\mathrm{uni}}\log(|\Gcal|/\delta)}{n}
}.
$$
\end{theorem}
Theorem~\ref{thm:fast_rate_sc} (see Appendix~\ref{app:proof_minimax} for the detailed proof) establishes the first $\otil(1/n)$ statistical rate for offline learning in KL-regularized contextual zero-sum games, improving upon the $\otil(1/\sqrt n)$ rates achieved by prior pessimism-based methods~\citep{ye2024online}. 
More importantly, this fast rate is obtained \emph{without} any form of pessimism or uncertainty penalization. 
Unlike existing approaches that rely on conservative value estimates or worst-case optimization over version spaces, our analysis leverages a fundamentally different mechanism based on stability of KL-regularized equilibria together with unilateral concentrability. This establishes the first provably correct statistical framework for KL-regularized games that eliminates the need for pessimism.

\section{Offline Self-play Policy Optimization}\label{sec:self_play}
While the minimax formulation in Algorithm~\ref{alg:minimax} provides a clean statistical characterization,
directly solving the KL-regularized saddle-point problem may be computationally expensive in large-scale contextual settings.
In this section, we turn to the computational aspect and present an efficient self-play policy optimization algorithm for approximately solving the empirical KL-regularized game induced by $\widehat g$.
\subsection{Algorithm}
Our algorithm is summarized in Algorithm~\ref{alg:self_play}. It follows a two-stage structure. We first compute the least-squares estimator $\widehat g$ from the offline dataset. Then, fixing $\widehat g$, we run a KL-regularized self-play procedure to approximately solve the induced zero-sum game. At each iteration, each player forms its expected payoff against the opponent's current policy and performs a KL-regularized mirror-descent update.
This yields a coupled mirror descent-ascent dynamics that drives the policy pair toward the Nash equilibrium of the empirical regularized game defined by $\widehat g$.

\begin{algorithm*}[t]
\caption{Offline Self-Play Policy Optimization}
\label{alg:self_play}
\begin{algorithmic}[1]

\STATE \textbf{Input:} coefficient $\eta>0$, learning rate schedule $\{\alpha_t\}_{t=0}^{T-1}$, number of iterations $T$, offline dataset $\mathcal{D}$, reference policy $\piref$, function class $\mathcal{G}$.

\STATE \textbf{Payoff estimation:} compute the least-squares estimator
\[
\wh{g} \in \arg\min_{g\in\mathcal{G}} 
\sum_{(x_i,a_{i,1},a_{i,2},p_i)\in\mathcal{D}}
\big(g(x_i,a_{i,1},a_{i,2})-p_i\big)^2 .
\]

\STATE \textbf{Initialization:} set $\pi_1^{(0)}(\cdot|x)=\piref(\cdot|x)$ and
$\pi_2^{(0)}(\cdot|x)=\piref(\cdot|x)$ for all contexts $x$.

\FOR{$t=0,1,\dots,T-1$}

\STATE Compute the payoff vectors against the opponent:
\[
\wh{f}_1^{(t)}(x,a_1)=\E_{a_2\sim \pi_2^{(t)}(\cdot|x)}[\wh{g}(x,a_1,a_2)],\qquad
\wh{f}_2^{(t)}(x,a_2)=\E_{a_1\sim \pi_1^{(t)}(\cdot|x)}[\wh{g}(x,a_1,a_2)].
\]

\STATE Update Player~1:
\[
\pi_1^{(t+1)}(a|x)\propto
\pi_1^{(t)}(a|x)^{1-\alpha_t\eta^{-1}}
\exp\!\big(\alpha_t\,\wh{f}_1^{(t)}(x,a)\big)\,
\piref(a|x)^{\alpha_t\eta^{-1}} .
\]

\STATE Update Player~2:
\[
\pi_2^{(t+1)}(a|x)\propto
\pi_2^{(t)}(a|x)^{1-\alpha_t\eta^{-1}}
\exp\!\big(-\alpha_t\,\wh{f}_2^{(t)}(\!x,a)\big)\,
\piref(a|x)^{\alpha_t\eta^{-1}} .
\]

\ENDFOR

\STATE \textbf{Output:} final policy pair $(\pi_1^{(T)},\pi_2^{(T)})$.

\end{algorithmic}
\end{algorithm*}

\subsection{Analysis}
We analyze Algorithm~\ref{alg:self_play} by bounding the duality gap of its last iterate; the full proofs for the lemmas in this subsection are deferred to Appendix~\ref{app:proof_self}. 

The key observation is that we can analyze the last iterate through a local transfer argument around the empirical equilibrium $\widehat\pi$:
\[
\mathrm{DualGap}(\pi^{(T)})
\le
\mathrm{DualGap}(\widehat\pi)
+
\Big|
\mathrm{DualGap}(\pi^{(T)})
-
\mathrm{DualGap}(\widehat\pi)
\Big|.
\]
The first term is the statistical error already controlled in Section~\ref{sec:info}. It remains to upper bound the second term, which captures the additional optimization error introduced by running self-play instead of computing the empirical equilibrium exactly.

We first establish a boundedness property of the KL-regularized dynamics. The update in Algorithm~\ref{alg:self_play} preserves a log-linear structure relative to the reference policy, and hence all iterates remain in a bounded log-density-ratio class.

\begin{lemma}[Log-linear Boundedness of Self-play Iterates]
\label{lem:log_linear_bounded}
Let $\{\pi^{(t)}\}_{t=0}^T$ be generated by Algorithm~\ref{alg:self_play} with $\alpha_t\le \eta$ for all $t$.
Then, for every $t\ge 0$ and each player $i\in\{1,2\}$,
\[
\sup_{x}\sup_{a\in\supp(\piref(\cdot|x))}
\left|
\log\frac{\pi_i^{(t)}(a|x)}{\piref(a|x)}
\right|
\le 2\eta .
\]
Moreover, the empirical Nash equilibrium $\widehat\pi$ induced by $\widehat g$ also satisfies this bound.
\end{lemma}
Lemma~\ref{lem:log_linear_bounded} shows that both the optimization trajectory and the target empirical equilibrium stay in the same bounded policy class. This boundedness is important because it makes the KL regularization term uniformly Lipschitz with respect to policy perturbations, without requiring any lower bound on the minimum action probability. We next show that, over this bounded log-linear class, the duality gap changes smoothly with the policy pair.

\begin{lemma}[Lipschitzness of the Duality Gap]
\label{lem:gap_lipschitz}
Let $\pi=(\pi_1,\pi_2)$ and $\pi'=(\pi_1',\pi_2')$ be two policy pairs.
Suppose that for each player $i\in\{1,2\}$ and every context $x$,
\begin{equation*}
\sup_{a\in\supp(\piref(\cdot|x))}
\max\left\{
\left|\log\frac{\pi_i(a|x)}{\piref(a|x)}\right|,
\left|\log\frac{\pi_i'(a|x)}{\piref(a|x)}\right|
\right\}
\le 2\eta .
\end{equation*}
Then, we have
\begin{equation}
\begin{aligned}
\left|
\mathrm{DualGap}(\pi)
-
\mathrm{DualGap}(\pi')
\right|
\le
3\,\mathbb E_{x\sim\rho}
\Big[
&\|\pi_1(\cdot|x)-\pi_1'(\cdot|x)\|_1  +
\|\pi_2(\cdot|x)-\pi_2'(\cdot|x)\|_1
\Big]. \nonumber
\end{aligned}
\end{equation}
\end{lemma}
Lemma~\ref{lem:gap_lipschitz} allows us to compare the duality gap of the last iterate with that of the empirical equilibrium directly.
Combined with Lemma~\ref{lem:log_linear_bounded}, it yields
\[
\mathrm{DualGap}(\pi^{(T)})
\le
\mathrm{DualGap}(\widehat\pi)
+
3\,\mathbb E_{x\sim\rho}
\left[
\|\pi_1^{(T)}(\cdot|x)-\widehat\pi_1(\cdot|x)\|_1
+
\|\pi_2^{(T)}(\cdot|x)-\widehat\pi_2(\cdot|x)\|_1
\right].
\]
Therefore, it remains to control the distance between the last iterate and the empirical equilibrium, which is achieved by the optimization convergence lemma below.
\begin{lemma}[Optimization Error Convergence]
\label{lem:opt_error}
Let $\widehat{\pi} = (\widehat{\pi}_1, \widehat{\pi}_2)$ be the Nash equilibrium of the regularized empirical game defined by $\widehat g$.
With the time-varying learning rate $\alpha_t = \frac{2\eta}{t+2}$, the output $\pi^{(T)}$ of Algorithm~\ref{alg:self_play} satisfies
\begin{align*}
    \mathbb{E}_{x \sim \rho}
    \!\left[
    \KL(\widehat{\pi}_1(\cdot|x) \| \pi_1^{(T)}(\cdot|x))
    +
    \KL(\widehat{\pi}_2(\cdot|x) \| \pi_2^{(T)}(\cdot|x))
    \right]
    \le
    \frac{16\eta^2}{T+1}.
\end{align*}
Consequently, by Pinsker's inequality and Jensen's inequality, 
\[
\mathbb E_{x\sim\rho}
\left[
\|\pi_1^{(T)}(\cdot|x)-\widehat\pi_1(\cdot|x)\|_1
+
\|\pi_2^{(T)}(\cdot|x)-\widehat\pi_2(\cdot|x)\|_1
\right]
\le
\frac{8\eta}{\sqrt{T+1}} .
\]
\end{lemma}
The proof views self-play as context-wise entropic OMD on the empirical KL-regularized game: applying the OMD one-step inequality yields a first-order term that contracts by payoff cancellation and the saddle-point property of $\widehat\pi$, plus a second-order term controlled by Lemma~\ref{lem:log_linear_bounded}. With $\alpha_t=2\eta/(t+2)$, this gives the last-iterate KL convergence after averaging over $x\sim\rho$, and the $\ell_1$ bound follows from Pinsker's and Jensen's inequalities.
\subsection{Theoretical Guarantee}
We now combine the transfer argument above with the statistical guarantee for the empirical equilibrium established in Section~\ref{sec:info}.
The full proof is provided in Appendix~\ref{app:proof_self}.

\begin{theorem}[Guarantee for Offline Self-Play]
\label{thm:self_play}
Under Assumptions \ref{asm:realize} and \ref{assm:uni_cover},
let $\pi^{(T)}$ be the output of Algorithm~\ref{alg:self_play} with learning rate
$\alpha_t = \frac{2\eta}{t+2}$.
Then with probability at least $1-\delta$, the duality gap satisfies
\[
\mathrm{DualGap}(\pi^{(T)})
\le
\mathcal{O}\!\left(
\frac{\eta}{\sqrt{T+1}}
+
\frac{(\eta+\eta^3)\,C_{\mathrm{uni}}\log(|\mathcal{G}|/\delta)}{n}
\right).
\]
\end{theorem}
The theorem shows that the last iterate of self-play inherits the fast statistical guarantee of the empirical equilibrium, up to an optimization error of order $\mathcal{O}(\eta/\sqrt{T})$.
Thus, self-play provides an efficient iterative approximation to exact empirical equilibrium computation, approaching its pessimism-free statistical guarantee as the optimization error decreases.

\section{Related Literature}\label{sec:related_work}
\paragraph{Offline Reinforcement Learning.}
Offline reinforcement learning (RL) studies the problem of learning decision-making policies from a fixed dataset collected by a behavior policy, without further interaction.
A fundamental difficulty in this setting is distribution shift: the target policy may induce state--action distributions that are insufficiently covered by the offline data, leading to extrapolation error and overestimation~\citep{levine2020offline}.
To address this challenge, the principle of \emph{pessimism} has become the standard tool for obtaining provable guarantees in offline RL~\citep{liu2020provably,rashidinejad2021bridging,jin2021pessimism,xie2021bellman,uehara2021pessimistic,zhan2022offline}.
Pessimistic methods penalize values of poorly supported state--action pairs, typically through explicit uncertainty bonuses or worst-case value constructions, and this framework is now well understood and known to be minimax optimal under single-policy concentrability assumptions~\citep{li2024settling}.
Beyond the single-agent setting, recent work has extended offline learning to multi-agent and game-theoretic environments, where equilibrium strategies must be learned from a fixed dataset.
Compared to single-agent offline RL, the central challenge in these settings arises from \emph{strategic} distribution shift induced by unilateral deviations of individual agents.
As a result, provably efficient learning generally requires stronger coverage conditions, formalized through notions such as \emph{unilateral concentrability}, which has been shown to be necessary in the worst case~\citep{cui2022offline}.
Existing approaches in multi-agent offline learning largely build on pessimism-based techniques, constructing conservative value estimates to control distribution shift~\citep{cui2022offline,zhong2022pessimistic,cui2022provably,zhang2023offline}.

\paragraph{KL-Regularized Objectives.}
KL regularization with respect to a reference policy has become a central tool for incorporating prior knowledge and enforcing behavioral constraints in reinforcement learning~\citep{xiong2023iterative,munos2024nash}.
By penalizing deviations from a reference distribution, such as a pre-trained model or human demonstrations, KL regularization reduces policy drift and stabilizes learning.
A prominent example is Reinforcement Learning from Human Feedback (RLHF), where a KL penalty is added to the reward to prevent the aligned model from deviating excessively from the pre-trained language model~\citep{ouyang2022training}.
Related KL-regularized formulations have been analyzed in both single-agent and multi-agent settings~\citep{xiong2023iterative,xie2024exploratory,zhao2025logarithmic,munos2024nash,ye2024online,zhang2025improving,zhang2025iterative,nayak2025achieving, chen2026Potential, chen2026pessimism}, where KL regularization is shown to limit policy deviation and improve stability of learning dynamics~\citep{ye2024online,nayak2025achieving}.
However, despite their empirical success, existing theoretical analyses of KL-regularized methods typically yield convergence or sample complexity rates that match those of their unregularized counterparts~\citep{ye2024online}.
Thus, it remains unclear whether KL regularization can lead to improved statistical rates in offline learning.

\paragraph{Learning in Games.}
Learning and equilibrium computation in static games, such as normal-form and contextual games, are classical topics in game theory. In particular, two-player zero-sum games have long served as a fundamental model, dating back to the seminal work of \citet{v1928theorie}.
\citet{freund1999adaptive} established a connection between no-regret online learning and equilibrium computation in zero-sum games, motivating a large body of subsequent work on using no-regret learning to solve stationary games
\citep{rakhlin2013optimization,daskalakis2011near,syrgkanis2015fast,chen2020hedging}.
These approaches show that when players repeatedly interact and update their strategies using no-regret algorithms such as multiplicative weights or mirror descent, the average play converges to Nash equilibria in zero-sum games and to coarse correlated equilibria more generally~\citep{roughgarden2016twenty}.
However, this line of work typically studies equilibrium computation through interactive no-regret learning and does not aim to provide statistical sample complexity guarantees for learning equilibria from a fixed offline dataset.

\section{Conclusion}\label{sec:conclusion}
In this paper, we develop a pessimism-free framework for offline learning in KL-regularized zero-sum games, showing that the geometry induced by KL regularization yields a fast $\widetilde{\mathcal O}(1/n)$ sample complexity bound under unilateral coverage.
We further provide a computationally efficient self-play algorithm that approximates the empirical equilibrium with a provable last-iterate optimization guarantee.
An important direction for future work is to extend this framework beyond zero-sum settings to more general game formulations.

\bibliography{ref}

@misc{chen2026pessimism,
      title={Pessimism-Free Offline Learning in General-Sum Games via KL Regularization}, 
      author={Claire Chen and Yuheng Zhang},
      year={2026},
      eprint={2605.00264},
      archivePrefix={arXiv},
      primaryClass={cs.LG},
      url={https://arxiv.org/abs/2605.00264}, 
}

@misc{chen2026Potential,
      title={Fast Rates in $\alpha$-Potential Games via Regularized Mirror Descent}, 
      author={Claire Chen and Yuheng Zhang},
      year={2026},
      eprint={2605.00268},
      archivePrefix={arXiv},
      primaryClass={cs.GT},
      url={https://arxiv.org/abs/2605.00268}, 
}

@article{nemirovski2004prox,
  title={Prox-method with rate of convergence O (1/t) for variational inequalities with Lipschitz continuous monotone operators and smooth convex-concave saddle point problems},
  author={Nemirovski, Arkadi},
  journal={SIAM Journal on Optimization},
  volume={15},
  number={1},
  pages={229--251},
  year={2004},
  publisher={SIAM}
}

@book{cover1999elements,
  title={Elements of information theory},
  author={Cover, Thomas M},
  year={1999},
  publisher={John Wiley \& Sons}
}

@book{cesa2006prediction,
  title={Prediction, learning, and games},
  author={Cesa-Bianchi, Nicolo and Lugosi, G{\'a}bor},
  year={2006},
  publisher={Cambridge university press}
}

@article{chen2020hedging,
  title={Hedging in games: Faster convergence of external and swap regrets},
  author={Chen, Xi and Peng, Binghui},
  journal={Advances in Neural Information Processing Systems},
  volume={33},
  pages={18990--18999},
  year={2020}
}

@article{syrgkanis2015fast,
  title={Fast convergence of regularized learning in games},
  author={Syrgkanis, Vasilis and Agarwal, Alekh and Luo, Haipeng and Schapire, Robert E},
  journal={Advances in Neural Information Processing Systems},
  volume={28},
  year={2015}
}

@inproceedings{daskalakis2011near,
  title={Near-optimal no-regret algorithms for zero-sum games},
  author={Daskalakis, Constantinos and Deckelbaum, Alan and Kim, Anthony},
  booktitle={Proceedings of the twenty-second annual ACM-SIAM symposium on Discrete Algorithms},
  pages={235--254},
  year={2011},
  organization={SIAM}
}

@article{rakhlin2013optimization,
  title={Optimization, learning, and games with predictable sequences},
  author={Rakhlin, Sasha and Sridharan, Karthik},
  journal={Advances in Neural Information Processing Systems},
  volume={26},
  year={2013}
}

@article{v1928theorie,
  title={Zur theorie der gesellschaftsspiele},
  author={v. Neumann, J},
  journal={Mathematische annalen},
  volume={100},
  number={1},
  pages={295--320},
  year={1928},
  publisher={Springer}
}

@book{roughgarden2016twenty,
  title={Twenty lectures on algorithmic game theory},
  author={Roughgarden, Tim},
  year={2016},
  publisher={Cambridge University Press}
}

@article{freund1999adaptive,
  title={Adaptive game playing using multiplicative weights},
  author={Freund, Yoav and Schapire, Robert E},
  journal={Games and Economic Behavior},
  volume={29},
  number={1-2},
  pages={79--103},
  year={1999},
  publisher={Elsevier}
}

@inproceedings{
zhang2025iterative,
title={Iterative Nash Policy Optimization: Aligning {LLM}s with General Preferences via No-Regret Learning},
author={Yuheng Zhang and Dian Yu and Baolin Peng and Linfeng Song and Ye Tian and Mingyue Huo and Nan Jiang and Haitao Mi and Dong Yu},
booktitle={The Thirteenth International Conference on Learning Representations},
year={2025}
}

@inproceedings{zhong2022pessimistic,
  title={Pessimistic minimax value iteration: Provably efficient equilibrium learning from offline datasets},
  author={Zhong, Han and Xiong, Wei and Tan, Jiyuan and Wang, Liwei and Zhang, Tong and Wang, Zhaoran and Yang, Zhuoran},
  booktitle={International Conference on Machine Learning},
  pages={27117--27142},
  year={2022},
  organization={PMLR}
}

@article{cui2022offline,
  title={When are offline two-player zero-sum Markov games solvable?},
  author={Cui, Qiwen and Du, Simon S},
  journal={Advances in Neural Information Processing Systems},
  volume={35},
  pages={25779--25791},
  year={2022}
}

@inproceedings{jin2021pessimism,
  title={Is pessimism provably efficient for offline rl?},
  author={Jin, Ying and Yang, Zhuoran and Wang, Zhaoran},
  booktitle={International conference on machine learning},
  pages={5084--5096},
  year={2021},
  organization={PMLR}
}

@article{xie2021bellman,
  title={Bellman-consistent pessimism for offline reinforcement learning},
  author={Xie, Tengyang and Cheng, Ching-An and Jiang, Nan and Mineiro, Paul and Agarwal, Alekh},
  journal={Advances in neural information processing systems},
  volume={34},
  pages={6683--6694},
  year={2021}
}

@inproceedings{zhang2023offline,
  title={Offline learning in markov games with general function approximation},
  author={Zhang, Yuheng and Bai, Yu and Jiang, Nan},
  booktitle={International Conference on Machine Learning},
  pages={40804--40829},
  year={2023},
  organization={PMLR}
}

@inproceedings{xie2021batch,
    author = "Xie, Tengyang and Jiang, Nan",
    year = "2021",
    booktitle = "Proceedings of the International Conference on Machine Learning",
    title = "Batch value-function approximation with only realizability"
}

@article{levine2020offline,
    author = "Levine, Sergey and Kumar, Aviral and Tucker, George and Fu, Justin",
    year = "2020",
    journal = "ArXiv Preprint",
    title = "Offline Reinforcement Learning: Tutorial, Review, and Perspectives on Open Problems"
}

@misc{ye2024online,
    author = "Ye, Chenlu and Xiong, Wei and Zhang, Yuheng and Jiang, Nan and Zhang, Tong",
    title = "Online Iterative Reinforcement Learning from Human Feedback with General Preference Model",
    year = "2024",
    archiveprefix = "ArXiv",
    eprint = "2402.07314",
    primaryclass = "cs.LG"
}

@article{komorowski2018artificial,
    author = "Komorowski, Matthieu and Celi, Leo A and Badawi, Omar and Gordon, Anthony C and Faisal, A Aldo",
    title = "The artificial intelligence clinician learns optimal treatment strategies for sepsis in intensive care",
    year = "2018",
    journal = "Nature Medicine"
}

@inproceedings{ouyang2022training,
    author = "Ouyang, Long and Wu, Jeffrey and Jiang, Xu and Almeida, Diogo and Wainwright, Carroll and Mishkin, Pamela and Zhang, Chong and Agarwal, Sandhini and Slama, Katarina and Ray, Alex and others",
    year = "2022",
    booktitle = "Advances in Neural Information Processing Systems",
    title = "Training language models to follow instructions with human feedback"
}

@article{cui2022provably,
  title={Provably efficient offline multi-agent reinforcement learning via strategy-wise bonus},
  author={Cui, Qiwen and Du, Simon S},
  journal={Advances in Neural Information Processing Systems},
  volume={35},
  pages={11739--11751},
  year={2022}
}

@article{uehara2021pessimistic,
  title={Pessimistic model-based offline reinforcement learning under partial coverage},
  author={Uehara, Masatoshi and Sun, Wen},
  journal={arXiv preprint arXiv:2107.06226},
  year={2021}
}

@article{li2024settling,
  title={Settling the sample complexity of model-based offline reinforcement learning},
  author={Li, Gen and Shi, Laixi and Chen, Yuxin and Chi, Yuejie and Wei, Yuting},
  journal={The Annals of Statistics},
  volume={52},
  number={1},
  pages={233--260},
  year={2024},
  publisher={Institute of Mathematical Statistics}
}

@article{liu2020provably,
  title={Provably good batch off-policy reinforcement learning without great exploration},
  author={Liu, Yao and Swaminathan, Adith and Agarwal, Alekh and Brunskill, Emma},
  journal={Advances in neural information processing systems},
  volume={33},
  pages={1264--1274},
  year={2020}
}

@article{rashidinejad2021bridging,
  title={Bridging offline reinforcement learning and imitation learning: A tale of pessimism},
  author={Rashidinejad, Paria and Zhu, Banghua and Ma, Cong and Jiao, Jiantao and Russell, Stuart},
  journal={Advances in Neural Information Processing Systems},
  volume={34},
  pages={11702--11716},
  year={2021}
}

@inproceedings{zhan2022offline,
  title={Offline reinforcement learning with realizability and single-policy concentrability},
  author={Zhan, Wenhao and Huang, Baihe and Huang, Audrey and Jiang, Nan and Lee, Jason},
  booktitle={Conference on Learning Theory},
  pages={2730--2775},
  year={2022},
  organization={PMLR}
}

@article{xiong2023iterative,
  title={Iterative preference learning from human feedback: Bridging theory and practice for rlhf under kl-constraint},
  author={Xiong, Wei and Dong, Hanze and Ye, Chenlu and Wang, Ziqi and Zhong, Han and Ji, Heng and Jiang, Nan and Zhang, Tong},
  journal={arXiv preprint arXiv:2312.11456},
  year={2023}
}

@article{xie2024exploratory,
  title={Exploratory preference optimization: Harnessing implicit q*-approximation for sample-efficient rlhf},
  author={Xie, Tengyang and Foster, Dylan J and Krishnamurthy, Akshay and Rosset, Corby and Awadallah, Ahmed and Rakhlin, Alexander},
  journal={arXiv preprint arXiv:2405.21046},
  year={2024}
}

@article{zhao2025logarithmic,
  title={Logarithmic regret for online kl-regularized reinforcement learning},
  author={Zhao, Heyang and Ye, Chenlu and Xiong, Wei and Gu, Quanquan and Zhang, Tong},
  journal={arXiv preprint arXiv:2502.07460},
  year={2025}
}

@inproceedings{munos2024nash,
  title={Nash learning from human feedback},
  author={Munos, R{\'e}mi and Valko, Michal and Calandriello, Daniele and Azar, Mohammad Gheshlaghi and Rowland, Mark and Guo, Zhaohan Daniel and Tang, Yunhao and Geist, Matthieu and Mesnard, Thomas and Fiegel, Come and others},
  booktitle={Forty-first International Conference on Machine Learning},
  year={2024}
}

@article{zhang2025improving,
  title={Improving LLM general preference alignment via optimistic online mirror descent},
  author={Zhang, Yuheng and Yu, Dian and Ge, Tao and Song, Linfeng and Zeng, Zhichen and Mi, Haitao and Jiang, Nan and Yu, Dong},
  journal={arXiv preprint arXiv:2502.16852},
  year={2025}
}

@article{nayak2025achieving,
  title={Achieving Logarithmic Regret in KL-Regularized Zero-Sum Markov Games},
  author={Nayak, Anupam and Yang, Tong and Yagan, Osman and Joshi, Gauri and Chi, Yuejie},
  journal={arXiv preprint arXiv:2510.13060},
  year={2025}
}

@article{zhao2025towards,
  title={Towards a Sharp Analysis of Offline Policy Learning for $ f $-Divergence-Regularized Contextual Bandits},
  author={Zhao, Qingyue and Ji, Kaixuan and Zhao, Heyang and Zhang, Tong and Gu, Quanquan},
  journal={arXiv preprint arXiv:2502.06051},
  year={2025}
}

@article{lee2023offline,
  title={Offline reinforcement learning for automated stock trading},
  author={Lee, Namyeong and Moon, Jun},
  journal={IEEE Access},
  volume={11},
  pages={112577--112589},
  year={2023},
  publisher={IEEE}
}

@article{shi2021offline,
  title={Offline reinforcement learning for autonomous driving with safety and exploration enhancement},
  author={Shi, Tianyu and Chen, Dong and Chen, Kaian and Li, Zhaojian},
  journal={arXiv preprint arXiv:2110.07067},
  year={2021}
}

\newpage
\appendix
\onecolumn

\section{Concentration Analysis for Payoff Estimation}\label{app:proof_conc}

In this section, we provide the concentration analysis for the least-squares payoff estimator $\widehat{g}$ defined in Algorithm~\ref{alg:minimax} and Algorithm~\ref{alg:self_play}. We establish a fast $\mathcal{O}(1/n)$ convergence rate for the mean-squared error under the data generating distribution.

\begin{lemma}\label{lem:gen}
For any policies $\pi_1,\pi_2$, let $\{(z_i)\}_{i=1}^n = \{(x_i, a_{i,1},a_{i,2})\}_{i=1}^n$ be samples generated i.i.d. from $\mu = \rho \times \pi_1 \times \pi_2$.
With probability at least $1-\delta$, for all $g_1, g_2 \in \mathcal{G}$, we have
\begin{align}
    \mathbb{E}_{z \sim \mu}\big[\big(g_1(z) - g_2(z)\big)^2 \big] \leq \frac{2}{n}\sum_{i=1}^n \big(g_1(z_i) - g_2(z_i)\big)^2 + \frac{80}{3n}\log(2|\mathcal{G}|/\delta). \nonumber
\end{align}
\end{lemma}
\begin{proof}
Let $X(z) = (g_1(z) - g_2(z))^2$. Since $g_1, g_2 \in [-1, 1]$, we have $X(z) \in [0, 4]$. We apply Bernstein's inequality. For a fixed pair $g_1, g_2$, with probability at least $1 - \delta'$, we have:
\begin{align}
    \mathbb{E}[X] - \frac{1}{n}\sum_{i=1}^n X(z_i) \leq \sqrt{\frac{2 \text{Var}(X) \log(1/\delta')}{n}} + \frac{8\log(1/\delta')}{3n}. \nonumber
\end{align}
Since $X(z) \in [0, 4]$, we have $\text{Var}(X) \leq \mathbb{E}[X^2] \leq 4 \mathbb{E}[X]$. Substituting this variance bound:
\begin{align}
    \mathbb{E}[X] - \frac{1}{n}\sum_{i=1}^n X(z_i) \leq \sqrt{\frac{8 \mathbb{E}[X] \log(1/\delta')}{n}} + \frac{8 \log(1/\delta')}{3n}. \nonumber
\end{align}
Using the AM-GM inequality $\sqrt{xy} \leq \frac{x}{2} + \frac{y}{2}$ with $x = \mathbb{E}[X]$ and $y = \frac{8 \log(1/\delta')}{n}$, we bound the square root term:
\begin{align}
    \sqrt{\frac{8 \mathbb{E}[X] \log(1/\delta')}{n}} \leq \frac{1}{2}\mathbb{E}[X] + \frac{4 \log(1/\delta')}{n}. \nonumber
\end{align}
Plugging this back in:
\begin{align}
    \mathbb{E}[X] - \frac{1}{n}\sum_{i=1}^n X(z_i) &\leq \frac{1}{2}\mathbb{E}[X] + \frac{4 \log(1/\delta')}{n} + \frac{8 \log(1/\delta')}{3n} \nonumber \\
    \implies \frac{1}{2}\mathbb{E}[X] &\leq \frac{1}{n}\sum_{i=1}^n X(z_i) + \frac{20 \log(1/\delta')}{3n}. \nonumber
\end{align}
Multiplying by 2:
\begin{align}
    \mathbb{E}[X] \leq \frac{2}{n}\sum_{i=1}^n X(z_i) + \frac{40 \log(1/\delta')}{3n}. \nonumber
\end{align}
Finally, we apply a union bound over all pairs $(g_1, g_2) \in \mathcal{G} \times \mathcal{G}$. Setting $\delta' = \delta / |\mathcal{G}|^2$, we have:
\begin{align}
    \log(1/\delta') = \log(|\mathcal{G}|^2/\delta) \leq 2\log(2|\mathcal{G}|/\delta). \nonumber
\end{align}
Substituting this into the bound yields:
\begin{align}
    \mathbb{E}_{\mu}\big[\big(g_1 - g_2\big)^2 \big] \leq \frac{2}{n}\sum_{i=1}^n \big(g_1(z_i) - g_2(z_i)\big)^2 + \frac{80}{3n}\log(2|\mathcal{G}|/\delta). \nonumber
\end{align}
This holds for all $g_1, g_2 \in \mathcal{G}$ simultaneously with probability at least $1-\delta$.
\end{proof}

\begin{lemma}[In-sample Error Bound]\label{lem:mle_finite}
Under Assumption~\ref{asm:realize}, let $\widehat{g}$ be the least squares estimator defined in Algorithm~\ref{alg:minimax}. With probability at least $1-\delta$,
\[
\sum_{i=1}^n \big(\widehat{g}(z_i) - g^\star(z_i)\big)^2
\;\le\;
8 \log\!\left(\frac{|\mathcal{G}|}{\delta}\right).
\]
\end{lemma}

\begin{proof}
By the optimality of $\widehat{g}$ and the realizability $g^\star \in \mathcal{G}$, we have $\sum_i (\widehat{g}(z_i) - p_i)^2 \le \sum_i (g^\star(z_i) - p_i)^2$. Substituting $p_i = g^\star(z_i) + \epsilon_i$ yields the inequality:
\begin{align}
\sum_{i=1}^n \big(\widehat{g}(z_i) - g^\star(z_i)\big)^2
\;\le\;
2 \sum_{i=1}^n \epsilon_i \big(\widehat{g}(z_i) - g^\star(z_i)\big).
\label{eq:basic_ineq}
\end{align}
Fix any $g \in \mathcal{G}$ and let $f_i = g(z_i) - g^\star(z_i)$. Since $\epsilon_i$ is $1$-sub-Gaussian, the exponential moment satisfies
$\mathbb{E}[\exp(\frac{1}{2} \sum_i \epsilon_i f_i - \frac{1}{8} \sum_i f_i^2)] \le 1$.
Applying the Chernoff method and a union bound over $\mathcal{G}$, we have that with probability at least $1-\delta$, for all $g \in \mathcal{G}$:
\begin{align}
2 \sum_{i=1}^n \epsilon_i \big(g(z_i) - g^\star(z_i)\big)
\;\le\;
\frac{1}{2} \sum_{i=1}^n \big(g(z_i) - g^\star(z_i)\big)^2
+ 4 \log\!\left(\frac{|\mathcal{G}|}{\delta}\right). \nonumber
\end{align}
Plugging this bound for $\widehat{g}$ back into Eq.~\eqref{eq:basic_ineq} and rearranging terms concludes the proof.
\end{proof}

\begin{corollary}[Fast Rate for Payoff Estimation]\label{cor:fast_rate}
Under Assumption~\ref{asm:realize}, let $\mu = \rho \times \pi_b$ be the data generating distribution. With probability at least $1-\delta$,
\begin{align}
    \mathbb{E}_{z \sim \mu}\big[\big(\widehat{g}(z) - g^\star(z)\big)^2 \big] \leq \mathcal{O}\!\left( \frac{\log(|\mathcal{G}|/\delta)}{n} \right). \nonumber
\end{align}
\end{corollary}

\begin{proof}
Applying Lemma~\ref{lem:gen} with $g_1 = \widehat{g}$ and $g_2 = g^\star$, we have
\begin{align}
    \mathbb{E}_{\mu}\big[\big(\widehat{g} - g^\star\big)^2 \big] \leq \frac{2}{n}\sum_{i=1}^n \big(\widehat{g}(z_i) - g^\star(z_i)\big)^2 + \frac{80}{3n}\log(2|\mathcal{G}|/\delta). \nonumber
\end{align}
Applying Lemma~\ref{lem:mle_finite} to bound the summation term:
\begin{align}
    \frac{2}{n}\sum_{i=1}^n \big(\widehat{g}(z_i) - g^\star(z_i)\big)^2 \le \frac{16}{n} \log\!\left(\frac{|\mathcal{G}|}{\delta}\right). \nonumber
\end{align}
Combining these yields the result.
\end{proof}
\section{Proofs for Section~\ref{sec:info}}\label{app:proof_minimax}

In this section, we provide the full proofs for the sample complexity results of the pessimism-free minimax algorithm presented in Section~\ref{sec:info}.

\subsection{Proof of Lemma~\ref{lem:kl_subopt}}

\begin{proof}
We prove the result for Player 1; the case for Player 2 follows symmetrically. Fix a context $x$ and an opponent policy $\pi_2$. Let $h(a_1) = \mathbb{E}_{a_2 \sim \pi_2(\cdot|x)}[g^\star(x, a_1, a_2)]$. The local objective for Player 1 is $L(\pi_1) = \langle h, \pi_1 \rangle - \eta^{-1} \KL(\pi_1 \| \piref)$.

The best response $\pi_1^\dagger(\cdot|x)$ maximizes $L(\pi_1)$. The first-order optimality condition implies that for all $a_1$:
\begin{align}\label{eq:foc_concise}
h(a_1) = \eta^{-1} \left( \log \frac{\pi_1^\dagger(a_1|x)}{\piref(a_1|x)} + C \right),
\end{align}
where $C$ is a normalization constant.
Consider the suboptimality gap for any policy $\pi_1$:
\begin{align*}
L(\pi_1^\dagger) - L(\pi_1)
&= \langle h, \pi_1^\dagger - \pi_1 \rangle - \eta^{-1} \big( \KL(\pi_1^\dagger \| \piref) - \KL(\pi_1 \| \piref) \big).
\end{align*}
Substituting $h(a_1)$ from Eq.~\eqref{eq:foc_concise}, the term involving $C$ vanishes since $\langle C \mathbf{1}, \pi_1^\dagger - \pi_1 \rangle = 0$. The remaining terms simplify as:
\begin{align*}
L(\pi_1^\dagger) - L(\pi_1)
&= \eta^{-1} \sum_{a_1} (\pi_1^\dagger(a_1) - \pi_1(a_1)) \log \frac{\pi_1^\dagger(a_1)}{\piref(a_1)} \\
&\quad - \eta^{-1} \sum_{a_1} \left( \pi_1^\dagger(a_1) \log \frac{\pi_1^\dagger(a_1)}{\piref(a_1)} - \pi_1(a_1) \log \frac{\pi_1(a_1)}{\piref(a_1)} \right) \\
&= \eta^{-1} \sum_{a_1} \pi_1(a_1) \left( \log \frac{\pi_1(a_1)}{\piref(a_1)} - \log \frac{\pi_1^\dagger(a_1)}{\piref(a_1)} \right) \\
&= \eta^{-1} \KL(\pi_1(\cdot|x) \| \pi_1^\dagger(\cdot|x)).
\end{align*}
Aggregating over $x \sim \rho$ yields the result.
\end{proof}

\subsection{Proof of Lemma~\ref{lem:kl_logit}}
\begin{proof}
Let $\Phi(z) = \log \sum_{a} \exp(z_a)$ be the log-sum-exp function. Its gradient is $\nabla \Phi(z) = \Softmax(z)$. To prove the smoothness properties, we analyze the Hessian $\nabla^2 \Phi(z)$.
The entries of the Hessian are given by $[\nabla^2 \Phi(z)]_{ij} = P(i)\delta_{ij} - P(i)P(j)$, where $P = \Softmax(z)$. For any vector $v \in \mathbb{R}^{|\mathcal{A}|}$, the quadratic form is:
\begin{align*}
    v^\top \nabla^2 \Phi(z) v
    &= \sum_i P(i) v_i^2 - \left(\sum_i P(i) v_i\right)^2 \\
    &= \text{Var}_{i \sim P}[v_i] \le \mathbb{E}_{i \sim P}[v_i^2] \le \|v\|_\infty^2.
\end{align*}
This bound $v^\top \nabla^2 \Phi(z) v \le \|v\|_\infty^2$ implies that $\Phi$ is $1$-smooth with respect to the $\ell_\infty$-norm. The Bregman divergence generated by $\Phi$ satisfies
\[
    D_\Phi(z', z) = \Phi(z') - \Phi(z) - \langle \nabla \Phi(z), z' - z \rangle = \KL(\Softmax(z) \| \Softmax(z')).
\]
Using the $1$-smoothness property ($\Phi(z') \le \Phi(z) + \langle \nabla \Phi(z), z' - z \rangle + \frac{1}{2}\|z' - z\|_\infty^2$), we obtain:
\[
    \KL(\Softmax(z) \| \Softmax(z')) = D_\Phi(z', z) \le \frac{1}{2} \|z - z'\|_\infty^2.
\]
This proves Lemma~\ref{lem:kl_logit}.
\end{proof}

\subsection{Proof of Lemma~\ref{lem:l1_stability}}
\begin{proof}
Fix a context $x$. Let $\pi_x = (\pi_1(\cdot|x), \pi_2(\cdot|x))$ be the joint policy. We define the gradient operator $F_{g,x}(\pi_x)$ associated with the zero-sum game with payoff $g$:
\[
F_{g,x}(\pi_x) \coloneqq \begin{pmatrix} -\nabla_{\pi_1} J_x(g, \pi_1, \pi_2) \\ \nabla_{\pi_2} J_x(g, \pi_1, \pi_2) \end{pmatrix}
= \underbrace{\begin{pmatrix} 0 & -G_g(x) \\ G_g(x)^\top & 0 \end{pmatrix}}_{M_{g,x}} \pi_x + \eta^{-1} \begin{pmatrix} \nabla \KL_1 \\ \nabla \KL_2 \end{pmatrix},
\]
where $G_g(x)$ is the payoff matrix at $x$, and $\nabla \KL_i$ denotes the gradient of $\KL(\pi_i\|\piref)$ w.r.t. $\pi_i$.

The Nash equilibria $\pi^\star_x$ and $\wh{\pi}_x$ satisfy the first-order optimality conditions $F_{g^\star,x}(\pi^\star_x) = 0$ and $F_{\wh{g},x}(\wh{\pi}_x) = 0$.
The operator $F_{g,x}$ is strongly monotone due to the $\eta^{-1}$-strong convexity of the KL regularization. Since the linear part $M_{g,x}$ is skew-symmetric (i.e., $\langle M_{g,x} v, v \rangle = 0$ for any vector $v$), we have:
\begin{align}\label{eq:strong_mono_app}
\langle F_{g,x}(\pi) - F_{g,x}(\pi'), \pi - \pi' \rangle
&= \eta^{-1} \sum_{i=1}^2 \langle \nabla \KL_i(\pi_i) - \nabla \KL_i(\pi'_i), \pi_i - \pi'_i \rangle \nonumber \\
&\ge \frac{1}{2\eta} \|\pi - \pi'\|_1^2. \nonumber
\end{align}
Applying this to $\wh{\pi}_x$ and $\pi^\star_x$ with the operator $F_{g^\star,x}$:
\begin{align*}
\frac{1}{2\eta} \|\wh{\pi}_x - \pi^\star_x\|_1^2
&\le \langle F_{g^\star,x}(\wh{\pi}_x) - F_{g^\star,x}(\pi^\star_x), \wh{\pi}_x - \pi^\star_x \rangle \\
&= \langle F_{g^\star,x}(\wh{\pi}_x), \wh{\pi}_x - \pi^\star_x \rangle \quad (\text{since } F_{g^\star,x}(\pi^\star_x)=0).
\end{align*}
Subtracting $F_{\wh{g},x}(\wh{\pi}_x)=0$ from the inner product, the regularization terms cancel, leaving only the payoff difference:
\begin{align*}
\langle F_{g^\star,x}(\wh{\pi}_x) - F_{\wh{g},x}(\wh{\pi}_x), \wh{\pi}_x - \pi^\star_x \rangle
&= \langle (M_{g^\star,x} - M_{\wh{g},x})\wh{\pi}_x, \wh{\pi}_x - \pi^\star_x \rangle.
\end{align*}
Let $\Delta M_x = M_{g^\star,x} - M_{\wh{g},x}$. We decompose $\wh{\pi}_x = \pi^\star_x + (\wh{\pi}_x - \pi^\star_x)$ and exploit the skew-symmetry of $\Delta M_x$:
\begin{align*}
\langle \Delta M_x \wh{\pi}_x, \wh{\pi}_x - \pi^\star_x \rangle
= \langle \Delta M_x \pi^\star_x, \wh{\pi}_x - \pi^\star_x \rangle + \underbrace{\langle \Delta M_x (\wh{\pi}_x - \pi^\star_x), \wh{\pi}_x - \pi^\star_x \rangle}_{=0}.
\end{align*}
Combining inequalities and applying Hölder's inequality:
\[
\frac{1}{2\eta} \|\wh{\pi}_x - \pi^\star_x\|_1^2 \le \langle \Delta M_x \pi^\star_x, \wh{\pi}_x - \pi^\star_x \rangle \le \|\Delta M_x \pi^\star_x\|_\infty \|\wh{\pi}_x - \pi^\star_x\|_1.
\]
Dividing by $\|\wh{\pi}_x - \pi^\star_x\|_1$ yields $\|\wh{\pi}_x - \pi^\star_x\|_1 \le 2\eta \|\Delta M_x \pi^\star_x\|_\infty$.
The term $\|\Delta M_x \pi^\star_x\|_\infty$ corresponds precisely to the unilateral estimation error $\mathcal{E}(x)$.
\end{proof}

\subsection{Proof of Theorem~\ref{thm:fast_rate_sc}}

\begin{proof}
We follow the roadmap outlined in Section~\ref{sec:info}. From the smoothness of the KL-regularized best response (Lemma~\ref{lem:kl_logit}), we have the bound for $\mathrm{Gap}_1$:
\begin{align}\label{eq:app_gap1}
\mathrm{Gap}_1 \le \frac{\eta}{2} \E_{x \sim \rho} \big[ \|\wh{\pi}_2(\cdot|x) - \pi_2^\star(\cdot|x)\|_1^2 \big].
\end{align}
For $\mathrm{Gap}_2$, using the decomposition of logits $\wh{z} - z^\star$ into a policy deviation term and a unilateral estimation error term, and applying the inequality $(a+b)^2 \le 2a^2 + 2b^2$ alongside Lemma~\ref{lem:kl_logit}, we obtain:
\begin{align}\label{eq:app_gap2}
\mathrm{Gap}_2 \le 2\eta \E_{x \sim \rho} \big[ \|\wh{\pi}_1(\cdot|x) - \pi_1^\star(\cdot|x)\|_1^2 \big] + 2\eta \E_{x \sim \rho} \big[ \Ecal_{\mathrm{uni}, 1}(x)^2 \big],
\end{align}
where $\Ecal_{\mathrm{uni}, 1}(x) = \max_{a_2} | \E_{a_1 \sim \pi_1^\star}[\wh{g}(x,a_1,a_2) - g^\star(x,a_1,a_2)] |$.

\paragraph{Stability and Error Unification.}
Summing Eq.~\eqref{eq:app_gap1} and Eq.~\eqref{eq:app_gap2}, the one-sided suboptimality is bounded by the joint policy distance and the unilateral estimation error:
\begin{align*}
\mathrm{Gap}_1 + \mathrm{Gap}_2 \le \frac{5\eta}{2} \E_{x \sim \rho} \big[ \|\wh{\pi}(\cdot|x) - \pi^\star(\cdot|x)\|_1^2 \big] + 2\eta \E_{x \sim \rho} \big[ \Ecal(x)^2 \big],
\end{align*}
where $\Ecal(x)$ is the unilateral error defined in Lemma~\ref{lem:l1_stability}.
Applying the stability bound from Lemma~\ref{lem:l1_stability}, which established $\|\wh{\pi}_x - \pi_x^\star\|_1 \le 2\eta \Ecal(x)$, we substitute the policy distance term to obtain:
\begin{align*}
\mathrm{Gap}_1 + \mathrm{Gap}_2 &\le \frac{5\eta}{2} \E_{x \sim \rho} \big[ (2\eta \Ecal(x))^2 \big] + 2\eta \E_{x \sim \rho} \big[ \Ecal(x)^2 \big] \\
&= (10\eta^3 + 2\eta) \E_{x \sim \rho} \big[ \Ecal(x)^2 \big].
\end{align*}

\paragraph{Concentrability and Statistical Bound.}
We now bound the expected squared unilateral error $\E_{x \sim \rho} [\Ecal(x)^2]$. By Jensen's inequality and the definition of the $\ell_\infty$ norm, we have:
\[
\Ecal(x)^2 \le \max_{i \in \{1,2\}} \max_{\pi_i \in \Pi} \E_{a_i \sim \pi_i, a_{-i} \sim \pi_{-i}^\star} \big[ (\wh{g}(x,a_1,a_2) - g^\star(x,a_1,a_2))^2 \big].
\]
Taking distribution over $x \sim \rho$ and applying Assumption~\ref{assm:uni_cover} yields:
\begin{align*}
\E_{x \sim \rho} [\Ecal(x)^2] &\le \max_{i \in \{1,2\}} \max_{\pi_i \in \Pi} \E_{\rho \times \pi_i \times \pi_{-i}^\star} [(\wh{g} - g^\star)^2] \\
&\le C_{\mathrm{uni}} \E_{\rho \times \pi_b} [(\wh{g} - g^\star)^2].
\end{align*}
Finally, we invoke the concentration result for the payoff estimator. Under Assumption~\ref{asm:realize}, Corollary~\ref{cor:fast_rate} ensures that with probability at least $1-\delta$:
\[
\E_{\rho \times \pi_b} [(\wh{g} - g^\star)^2] \le \mathcal{O}\!\left( \frac{\log(|\Gcal|/\delta)}{n} \right).
\]
Combining these steps, we conclude:
\[
\mathrm{Gap}_1 + \mathrm{Gap}_2 \le \mathcal{O}\!\left( \frac{(\eta + \eta^3) C_{\mathrm{uni}} \log(|\Gcal|/\delta)}{n} \right).
\]
This completes the proof, establishing an $\otil(1/n)$ sample complexity.
\end{proof}

\section{Proofs for Section~\ref{sec:self_play}}
\label{app:proof_self}

In this section, we provide the full proofs for the self-play policy optimization results presented in Section~\ref{sec:self_play}.

\subsection{Proof of Lemma~\ref{lem:log_linear_bounded}}
Fix an arbitrary context $x$ and let $\mathcal A_x :=\supp(\piref(\cdot|x))$ denote the set of actions supported by the reference policy at context $x$. For any vector $v\in\mathbb R^{\mathcal A_x}$, we write $\operatorname{osc}(v)
:=
\max_{a\in\mathcal A_x} v(a)
-
\min_{a\in\mathcal A_x} v(a)
$ for its oscillation over $\mathcal A_x$.

\begin{proof}
For each player $i\in\{1,2\}$, define the log-density ratio
\[
u_i^{(t)}(a)
:=
\log\frac{\pi_i^{(t)}(a|x)}{\piref(a|x)},
\qquad a\in\mathcal A_x .
\]
We show by induction that
$\operatorname{osc}(u_i^{(t)})\le 2\eta$ for all $t\ge 0$ and both players.
At initialization, $\pi_i^{(0)}(\cdot|x)=\piref(\cdot|x)$, so
$u_i^{(0)}\equiv 0$ and the claim holds. Consider Player~1. The update in Algorithm~\ref{alg:self_play} can be written as
\[
u_1^{(t+1)}(a)
=
w_1^{(t)}(a)
-
\log Z_1^{(t)}(x),
\]
where
\[
w_1^{(t)}(a)
=
\left(1-\frac{\alpha_t}{\eta}\right)u_1^{(t)}(a)
+
\alpha_t \widehat f_1^{(t)}(x,a),
\qquad
Z_1^{(t)}(x)
=
\sum_{b\in\mathcal A_x}
\piref(b|x)\exp\!\big(w_1^{(t)}(b)\big).
\]
Since $\widehat g\in\mathcal G\subset[-1,1]$, we have
$\operatorname{osc}(\widehat f_1^{(t)}(x,\cdot))\le 2$.
Using $\alpha_t\le\eta$ and the induction hypothesis, we obtain
\[
\begin{aligned}
\operatorname{osc}\big(w_1^{(t)}\big)
&\le
\left(1-\frac{\alpha_t}{\eta}\right)
\operatorname{osc}\big(u_1^{(t)}\big)
+
\alpha_t
\operatorname{osc}\big(\widehat f_1^{(t)}(x,\cdot)\big) \\
&\le
\left(1-\frac{\alpha_t}{\eta}\right)2\eta
+
2\alpha_t
=
2\eta .
\end{aligned}
\]
Since subtracting a constant does not change oscillation,
\[
\operatorname{osc}\big(u_1^{(t+1)}\big)
=
\operatorname{osc}\big(w_1^{(t)}\big)
\le 2\eta .
\]
Moreover, because $\piref(\cdot|x)$ is a probability distribution on $\mathcal A_x$,
\[
\min_{a\in\mathcal A_x} w_1^{(t)}(a)
\le
\log Z_1^{(t)}(x)
\le
\max_{a\in\mathcal A_x} w_1^{(t)}(a).
\]
Hence,
\[
\left|u_1^{(t+1)}(a)\right|
\le
\operatorname{osc}\big(w_1^{(t)}\big)
\le
2\eta ,
\qquad
a\in\mathcal A_x .
\]

The argument for Player~2 is identical, with
$-\widehat f_2^{(t)}(x,\cdot)$ replacing
$\widehat f_1^{(t)}(x,\cdot)$.
Since the context $x$ was arbitrary, the induction establishes the following oscillation estimate for all iterates:
\begin{equation}
\label{eq:selfplay_log_ratio_osc}
\operatorname{osc}_{a\in\mathcal A_x}
\left(
\log\frac{\pi_i^{(t)}(a|x)}{\piref(a|x)}
\right)
\le 2\eta,
\qquad
i\in\{1,2\},\ t\ge 0 .
\end{equation}
The same argument also gives the stated absolute bound,
\[
\sup_x\sup_{a\in\supp(\piref(\cdot|x))}
\left|
\log\frac{\pi_i^{(t)}(a|x)}{\piref(a|x)}
\right|
\le 2\eta ,
\qquad i\in\{1,2\},\ t\ge 0 .
\]

It remains to show the same bound for the empirical Nash equilibrium
$\widehat\pi$.
For each context $x$, the KL-regularized best-response equations give
\[
\widehat\pi_1(a|x)
\propto
\piref(a|x)
\exp\!\left(
\eta\,
\mathbb E_{a_2\sim\widehat\pi_2(\cdot|x)}
[
\widehat g(x,a,a_2)
]
\right),
\]
and
\[
\widehat\pi_2(a|x)
\propto
\piref(a|x)
\exp\!\left(
-\eta\,
\mathbb E_{a_1\sim\widehat\pi_1(\cdot|x)}
[
\widehat g(x,a_1,a)
]
\right).
\]
In both cases, the unnormalized logit vector has oscillation at most
$2\eta$ because $\widehat g\in[-1,1]$.
The same normalization argument therefore implies
\[
\sup_x\sup_{a\in\supp(\piref(\cdot|x))}
\left|
\log\frac{\widehat\pi_i(a|x)}{\piref(a|x)}
\right|
\le 2\eta,
\qquad i\in\{1,2\}.
\]
This completes the proof.
\end{proof}
\subsection{Proof of Lemma~\ref{lem:gap_lipschitz}}

\begin{proof}
For a Player~2 policy $\pi_2$, define
$V_1(\pi_2):=J(g^\star,\dagger,\pi_2)$, and for a Player~1 policy $\pi_1$, define
$V_2(\pi_1):=J(g^\star,\pi_1,\dagger)$.
Then $\mathrm{DualGap}(\pi)=V_1(\pi_2)-V_2(\pi_1)$.
It suffices to show that $V_1$ is Lipschitz in $\pi_2$ and $V_2$ is Lipschitz in $\pi_1$ over the bounded log-density-ratio class.

We first consider $V_1$.
Since $V_1(\pi_2)=\max_{\widetilde\pi_1\in\Pi}J(g^\star,\widetilde\pi_1,\pi_2)$, for any two Player~2 policies $\pi_2$ and $\pi_2'$ satisfying the bounded log-density-ratio condition,
\[
|V_1(\pi_2)-V_1(\pi_2')|
\le
\sup_{\widetilde\pi_1\in\Pi}
|J(g^\star,\widetilde\pi_1,\pi_2)-J(g^\star,\widetilde\pi_1,\pi_2')|.
\]
Fix any $\widetilde\pi_1\in\Pi$ and any context $x$, since $g^\star\in[-1,1]$, we have
\[
\begin{aligned}
&\left|
\E_{a_1\sim\widetilde\pi_1(\cdot|x),\,a_2\sim\pi_2(\cdot|x)}
[g^\star(x,a_1,a_2)]
-
\E_{a_1\sim\widetilde\pi_1(\cdot|x),\,a_2\sim\pi_2'(\cdot|x)}
[g^\star(x,a_1,a_2)]
\right|\le
\|\pi_2(\cdot|x)-\pi_2'(\cdot|x)\|_1,
\end{aligned}
\]

It remains to bound the KL regularization term.
For fixed $x$, define
$\pi_{2,\lambda}(\cdot|x):=(1-\lambda)\pi_2(\cdot|x)+\lambda\pi_2'(\cdot|x)$ for $\lambda\in[0,1]$.
Since both $\pi_2$ and $\pi_2'$ satisfy the bounded log-density-ratio condition, so does $\pi_{2,\lambda}$. For every $a\in\mathcal A_x$,
\[
\frac{\pi_{2,\lambda}(a|x)}{\piref(a|x)}
=
(1-\lambda)\frac{\pi_2(a|x)}{\piref(a|x)}
+
\lambda\frac{\pi_2'(a|x)}{\piref(a|x)}
\in [e^{-2\eta},e^{2\eta}].
\]
Therefore, by the fundamental theorem of calculus,
\[
\begin{aligned}
&
\KL(\pi_2(\cdot|x)\|\piref(\cdot|x))
-
\KL(\pi_2'(\cdot|x)\|\piref(\cdot|x))
\\
&\quad =
\int_0^1
\sum_{a\in\mathcal A_x}
\big(\pi_2(a|x)-\pi_2'(a|x)\big)
\left(
\log\frac{\pi_{2,\lambda}(a|x)}{\piref(a|x)}+1
\right)
d\lambda .
\end{aligned}
\]
The constant term vanishes since
$\sum_{a\in\mathcal A_x}(\pi_2(a|x)-\pi_2'(a|x))=0$.
Using $\left|\log(\pi_{2,\lambda}(a|x)/\piref(a|x))\right|\le 2\eta$, we get
\[
\left|
\KL(\pi_2(\cdot|x)\|\piref(\cdot|x))
-
\KL(\pi_2'(\cdot|x)\|\piref(\cdot|x))
\right|
\le
2\eta\,
\|\pi_2(\cdot|x)-\pi_2'(\cdot|x)\|_1 .
\]
After multiplying by the regularization coefficient $\eta^{-1}$, the KL part contributes at most
$2\|\pi_2(\cdot|x)-\pi_2'(\cdot|x)\|_1$.
Combining the payoff and KL parts and taking expectation over $x\sim\rho$ gives
\[
|V_1(\pi_2)-V_1(\pi_2')|
\le
3\,\E_{x\sim\rho}
\|\pi_2(\cdot|x)-\pi_2'(\cdot|x)\|_1 .
\]
The argument for $V_2$ is symmetric and we obtain
\[
|V_2(\pi_1)-V_2(\pi_1')|
\le
3\,\E_{x\sim\rho}
\|\pi_1(\cdot|x)-\pi_1'(\cdot|x)\|_1 .
\]
Finally,
\[
\begin{aligned}
\left|
\mathrm{DualGap}(\pi)
-
\mathrm{DualGap}(\pi')
\right|
&=
\left|
V_1(\pi_2)-V_2(\pi_1)
-
V_1(\pi_2')+V_2(\pi_1')
\right|
\\
&\le
|V_1(\pi_2)-V_1(\pi_2')|
+
|V_2(\pi_1)-V_2(\pi_1')|
\\
&\le
3\,\mathbb E_{x\sim\rho}
\left[
\|\pi_1(\cdot|x)-\pi_1'(\cdot|x)\|_1
+
\|\pi_2(\cdot|x)-\pi_2'(\cdot|x)\|_1
\right].
\end{aligned}
\]
This completes the proof.
\end{proof}

\subsection{Proof of Lemma~\ref{lem:opt_error}}

We first introduce a helpful lemma for the online mirror descent update in the contextual setting.

\begin{lemma}
\label{lem:omd_lemma_contextual}
Fix a context $x$.
Let $\phi(\pi)=\sum_a \pi(a)\log\pi(a)$ be the negative entropy, and let
$\delta(x,\cdot)\in\mathbb R^{|\mathcal A_x|}$.
For any $\pi^-(\cdot|x)$, define
\[
\pi^+(\cdot|x)
=
\arg\max_{\pi\in\Delta(\mathcal A_x)}
\left[
\sum_a \pi(a)\delta(x,a)
-
\KL(\pi\|\pi^-(\cdot|x))
\right].
\]
Then for any $\pi\in\Delta(\mathcal A_x)$,
\[
\KL(\pi\|\pi^+(\cdot|x))
\le
\KL(\pi\|\pi^-(\cdot|x))
+
\sum_a(\pi^-(a|x)-\pi(a))\delta(x,a)
+
\frac{1}{8}\operatorname{osc}(\delta(x,\cdot))^2 .
\]
\end{lemma}
\begin{proof}
For notational simplicity, write $\delta(a)=\delta(x,a)$ and
$\pi^-(a)=\pi^-(a|x)$.
The closed-form solution of the update is
\[
\pi^+(a)
=
\frac{\pi^-(a)\exp(\delta(a))}
{\sum_{b\in\mathcal A_x}\pi^-(b)\exp(\delta(b))}.
\]
Therefore,
\[
\begin{aligned}
\KL(\pi\|\pi^+)-\KL(\pi\|\pi^-)
&=
\sum_{a\in\mathcal A_x}\pi(a)
\log\frac{\pi^-(a)}{\pi^+(a)}  \\
&=
-\sum_{a\in\mathcal A_x}\pi(a)\delta(a)
+
\log\sum_{a\in\mathcal A_x}\pi^-(a)\exp(\delta(a)).
\end{aligned}
\]
By Hoeffding's lemma applied to the bounded random variable
$\delta(a)$ with $a\sim\pi^-$,
\[
\log\sum_{a\in\mathcal A_x}\pi^-(a)\exp(\delta(a))
\le
\sum_{a\in\mathcal A_x}\pi^-(a)\delta(a)
+
\frac18\,\operatorname{osc}(\delta)^2.
\]
Combining the two equations gives
\[
\KL(\pi\|\pi^+)-\KL(\pi\|\pi^-)
\le
\sum_{a\in\mathcal A_x}(\pi^-(a)-\pi(a))\delta(a)
+
\frac18\,\operatorname{osc}(\delta)^2,
\]
which proves the claim.
\end{proof}
We now proceed to the proof of the optimization error bound.
\begin{proof}[Proof of Lemma~\ref{lem:opt_error}]
The proof proceeds by analyzing the iterates pointwise for each context $x$ and then taking the expectation.

\paragraph{1. Pointwise Analysis.}
Fix an arbitrary context $x$, and let $\mathcal A_x:=\supp(\piref(\cdot|x))$.
Define
\[
V_t(x)
=
\KL(\widehat{\pi}_1(\cdot|x)\|\pi_1^{(t)}(\cdot|x))
+
\KL(\widehat{\pi}_2(\cdot|x)\|\pi_2^{(t)}(\cdot|x)).
\]
We verify that the update rules in Algorithm~\ref{alg:self_play} correspond to the OMD update in Lemma~\ref{lem:omd_lemma_contextual} with the following directions:
\begin{align*}
    \delta_{1,t}(x,a)
    &=
    \alpha_t \widehat{f}_1^{(t)}(x,a)
    -
    \alpha_t \eta^{-1}
    \log \frac{\pi_1^{(t)}(a|x)}{\piref(a|x)}, \\
    \delta_{2,t}(x,a)
    &=
    -\alpha_t \widehat{f}_2^{(t)}(x,a)
    -
    \alpha_t \eta^{-1}
    \log \frac{\pi_2^{(t)}(a|x)}{\piref(a|x)} .
\end{align*}
Applying Lemma~\ref{lem:omd_lemma_contextual} with these directions to each player's update and summing the inequalities, we obtain
\begin{align}\label{eq:omd_1}
    V_{t+1}(x)
    \le
    V_t(x)
    +
    \underbrace{
    \sum_{i=1}^2
    \sum_{a\in\mathcal A_x}
    \big(\pi_i^{(t)}(a|x)-\widehat{\pi}_i(a|x)\big)
    \delta_{i,t}(x,a)
    }_{\mathrm{(I)}}
    +
    \underbrace{
    \frac18\operatorname{osc}(\delta_{1,t}(x,\cdot))^2
    +
    \frac18\operatorname{osc}(\delta_{2,t}(x,\cdot))^2
    }_{\mathrm{(II)}} .
\end{align}
For term $\mathrm{(II)}$, the boundedness of $\widehat g$ gives
$\operatorname{osc}(\widehat f_i^{(t)}(x,\cdot))\le 2$.
Moreover, Eq.~\eqref{eq:selfplay_log_ratio_osc} from the proof of Lemma~\ref{lem:log_linear_bounded} gives
\[
\operatorname{osc}\!\left(
\log\frac{\pi_i^{(t)}(\cdot|x)}{\piref(\cdot|x)}
\right)
\le 2\eta .
\]
Therefore,
\[
\operatorname{osc}(\delta_{i,t}(x,\cdot))
\le
\alpha_t\operatorname{osc}(\widehat f_i^{(t)}(x,\cdot))
+
\alpha_t\eta^{-1}
\operatorname{osc}\!\left(
\log\frac{\pi_i^{(t)}(\cdot|x)}{\piref(\cdot|x)}
\right)
\le
4\alpha_t .
\]
Thus, $\mathrm{(II)}\le 4\alpha_t^2$.

We now focus on analyzing term $\mathrm{(I)}$. Substituting the expressions for $\delta_{i,t}$ into Eq.~\eqref{eq:omd_1}, term $\mathrm{(I)}$ splits into two parts:
\begin{align*}
    \mathrm{(I)}
    &=
    \alpha_t
    \left[
    \langle \pi_1^{(t)}-\widehat{\pi}_1,\widehat f_1^{(t)}\rangle
    -
    \langle \pi_2^{(t)}-\widehat{\pi}_2,\widehat f_2^{(t)}\rangle
    \right] \\
    &\quad
    -
    \alpha_t\eta^{-1}
    \sum_{i=1}^2
    \left\langle
    \pi_i^{(t)}-\widehat{\pi}_i,
    \log\frac{\pi_i^{(t)}}{\piref}
    \right\rangle .
\end{align*}
Using the identity
\[
-
\left\langle
\pi_i^{(t)}-\widehat\pi_i,
\log\frac{\pi_i^{(t)}}{\piref}
\right\rangle
=
-\KL(\pi_i^{(t)}\|\piref)
+
\KL(\widehat\pi_i\|\piref)
-
\KL(\widehat\pi_i\|\pi_i^{(t)}),
\]
we can rewrite term $\mathrm{(I)}$ as
\begin{align*}
\mathrm{(I)}
&=
\alpha_t
\Bigg[
\langle \pi_1^{(t)}-\widehat{\pi}_1,\widehat f_1^{(t)}\rangle
-
\eta^{-1}\KL(\pi_1^{(t)}\|\piref)
+
\eta^{-1}\KL(\widehat\pi_1\|\piref)
\\
&\qquad
-
\left(
\langle \pi_2^{(t)}-\widehat{\pi}_2,\widehat f_2^{(t)}\rangle
+
\eta^{-1}\KL(\pi_2^{(t)}\|\piref)
-
\eta^{-1}\KL(\widehat\pi_2\|\piref)
\right)
\Bigg]
\\
&\quad
-
\alpha_t\eta^{-1}
\sum_{i=1}^2
\KL(\widehat\pi_i\|\pi_i^{(t)}).
\end{align*}
We analyze the term in the large brackets. Let $\widehat J_x$ be the regularized objective for the empirical game at context $x$:
\[
\widehat J_x(\pi_1,\pi_2)
:=
\mathbb E_{\substack{a_1\sim\pi_1(\cdot|x)\\ a_2\sim\pi_2(\cdot|x)}}
[\widehat g(x,a_1,a_2)]
-
\eta^{-1}\KL(\pi_1(\cdot|x)\|\piref(\cdot|x))
+
\eta^{-1}\KL(\pi_2(\cdot|x)\|\piref(\cdot|x)).
\]
Recall that
$\widehat f_1^{(t)}(x,a_1)
=
\mathbb E_{a_2\sim\pi_2^{(t)}(\cdot|x)}
[\widehat g(x,a_1,a_2)]$
and
$\widehat f_2^{(t)}(x,a_2)
=
\mathbb E_{a_1\sim\pi_1^{(t)}(\cdot|x)}
[\widehat g(x,a_1,a_2)]$.
First, the inner product terms involving the current policy pair cancel out exactly:
\[
\langle \pi_1^{(t)},\widehat f_1^{(t)}\rangle
=
\mathbb E_{\substack{a_1\sim\pi_1^{(t)}(\cdot|x)\\a_2\sim\pi_2^{(t)}(\cdot|x)}}
[\widehat g(x,a_1,a_2)]
=
\langle \pi_2^{(t)},\widehat f_2^{(t)}\rangle .
\]
Therefore, the large bracket equals $
\widehat J_x(\pi_1^{(t)},\widehat\pi_2)
-
\widehat J_x(\widehat\pi_1,\pi_2^{(t)}).$ Since the policy class contains all context-wise policies supported on $\piref$, the empirical regularized game decomposes over contexts. Hence the empirical Nash equilibrium satisfies the following pointwise saddle property for each context $x$:
\[
\widehat J_x(\pi_1^{(t)},\widehat\pi_2)
\le
\widehat J_x(\widehat\pi_1,\widehat\pi_2)
\le
\widehat J_x(\widehat\pi_1,\pi_2^{(t)}).
\]
Thus, the large bracket is non-positive, and hence
\[
\mathrm{(I)}
\le
-\alpha_t\eta^{-1}
\sum_{i=1}^2
\KL(\widehat\pi_i(\cdot|x)\|\pi_i^{(t)}(\cdot|x))
=
-\alpha_t\eta^{-1}V_t(x).
\]
Combining the bounds on $\mathrm{(I)}$ and $\mathrm{(II)}$ in Eq.~\eqref{eq:omd_1}, we obtain
\[
V_{t+1}(x)
\le
\left(1-\frac{\alpha_t}{\eta}\right)V_t(x)
+
4\alpha_t^2 .
\]

\paragraph{2. Solving the Recursion.}
With $\alpha_t=\frac{2\eta}{t+2}$, the recursion becomes
\[
V_{t+1}(x)
\le
\frac{t}{t+2}V_t(x)
+
\frac{16\eta^2}{(t+2)^2}.
\]
We prove $V_t(x)\le \frac{16\eta^2}{t+1}$ for all $t\ge 1$ by induction.
For the base case, the recursion with $t=0$ gives
\[
V_1(x)
\le
4\eta^2
\le
\frac{16\eta^2}{2}.
\]
Assume $V_t(x)\le \frac{16\eta^2}{t+1}$. Then
\begin{align*}
    V_{t+1}(x)
    &\le
    \frac{t}{t+2}\frac{16\eta^2}{t+1}
    +
    \frac{16\eta^2}{(t+2)^2} \\
    &=
    16\eta^2
    \frac{t^2+3t+1}{(t+1)(t+2)^2}
    \le
    \frac{16\eta^2}{t+2},
\end{align*}
where the last step uses $t^2+3t+1\le (t+1)(t+2)$.

\paragraph{3. Aggregation.}
Taking the expectation over $x\sim\rho$ yields
\[
\mathbb E_{x\sim\rho}[V_T(x)]
\le
\frac{16\eta^2}{T+1}.
\]
Finally, by Pinsker's inequality, for each $x$,
\[
\|\pi_1^{(T)}(\cdot|x)-\widehat\pi_1(\cdot|x)\|_1
+
\|\pi_2^{(T)}(\cdot|x)-\widehat\pi_2(\cdot|x)\|_1
\le
2\sqrt{V_T(x)}.
\]
Taking expectation over $x\sim\rho$ and applying Jensen's inequality gives
\[
\mathbb E_{x\sim\rho}
\left[
\|\pi_1^{(T)}(\cdot|x)-\widehat\pi_1(\cdot|x)\|_1
+
\|\pi_2^{(T)}(\cdot|x)-\widehat\pi_2(\cdot|x)\|_1
\right]
\le
2\sqrt{\mathbb E_{x\sim\rho}[V_T(x)]}
\le
\frac{8\eta}{\sqrt{T+1}}.
\]
This completes the proof.
\end{proof}

\subsection{Proof of Theorem~\ref{thm:self_play}}

\begin{proof}
We combine the local transfer argument, the optimization convergence of self-play, and the statistical guarantee for the empirical equilibrium.
Let $\widehat\pi=(\widehat\pi_1,\widehat\pi_2)$ denote the Nash equilibrium of the empirical KL-regularized game induced by $\widehat g$.

\paragraph{Bounding the Optimization Error.}
Since $\alpha_t=2\eta/(t+2)\le \eta$, Lemma~\ref{lem:log_linear_bounded} implies that both $\pi^{(T)}$ and $\widehat\pi$ satisfy the bounded log-density-ratio condition.
Therefore, by Lemma~\ref{lem:gap_lipschitz},
\begin{align}
\mathrm{DualGap}(\pi^{(T)})
&\le
\mathrm{DualGap}(\widehat\pi)
+
\left|
\mathrm{DualGap}(\pi^{(T)})
-
\mathrm{DualGap}(\widehat\pi)
\right| \nonumber\\
&\le
\mathrm{DualGap}(\widehat\pi)
+
3\,\mathbb E_{x\sim\rho}
\left[
\|\pi_1^{(T)}(\cdot|x)-\widehat\pi_1(\cdot|x)\|_1
+
\|\pi_2^{(T)}(\cdot|x)-\widehat\pi_2(\cdot|x)\|_1
\right].
\label{eq:self_play_transfer}
\end{align}
Lemma~\ref{lem:opt_error} further gives
\begin{align}
\mathbb E_{x\sim\rho}
\left[
\|\pi_1^{(T)}(\cdot|x)-\widehat\pi_1(\cdot|x)\|_1
+
\|\pi_2^{(T)}(\cdot|x)-\widehat\pi_2(\cdot|x)\|_1
\right]
\le
\frac{8\eta}{\sqrt{T+1}}.
\nonumber
\end{align}
Substituting this bound into Eq.~\eqref{eq:self_play_transfer}, the optimization error contributes at most $\mathcal O\!\left(\frac{\eta}{\sqrt{T+1}}\right)$.

\paragraph{Bounding the Statistical Error.}
It remains to control the true duality gap of the empirical equilibrium $\widehat\pi$.
By Theorem~\ref{thm:fast_rate_sc}, applied to the same least-squares estimator $\widehat g$, under Assumptions~\ref{asm:realize} and~\ref{assm:uni_cover}, with probability at least $1-\delta$,
\begin{align}
\mathrm{DualGap}(\widehat\pi)
\le
\mathcal O\!\left(
\frac{(\eta+\eta^3)\,C_{\mathrm{uni}}\log(|\mathcal G|/\delta)}{n}
\right).
\nonumber
\end{align}

\paragraph{Conclusion.}
Combining the optimization and statistical terms, we obtain
\begin{align}
\mathrm{DualGap}(\pi^{(T)})
\le
\mathcal O\!\left(
\frac{\eta}{\sqrt{T+1}}
+
\frac{(\eta+\eta^3)\,C_{\mathrm{uni}}\log(|\mathcal G|/\delta)}{n}
\right).
\nonumber
\end{align}
This completes the proof.
\end{proof}

\end{document}